\documentclass{ecai} 

\usepackage{latexsym}
\usepackage{amssymb}
\usepackage{amsmath}
\usepackage{amsthm}
\usepackage{booktabs}
\usepackage{enumitem}
\usepackage{graphicx}
\usepackage{color}
\usepackage{subcaption}

\usepackage{times}
\usepackage{soul}
\usepackage{url}
\usepackage[hidelinks]{hyperref}
\usepackage[utf8]{inputenc}

\usepackage{algorithm}
\usepackage[noend]{algorithmic}
\usepackage[switch]{lineno}

\usepackage{etoolbox} 
\usepackage{array}
\usepackage{eqparbox}

\newtheorem{theorem}{Theorem}
\newtheorem{lemma}[theorem]{Lemma}

\newtheorem{definition}{Definition}

\newtheorem{Claim}{Claim}
\newcommand{\Pbar}{\bar{\mathcal{P}}}
\newcommand{\MU}{\textit{Max-Util}}
\newcommand{\ME}{\textit{Max-Egal}}
\newcommand{\ALO}{\textit{At-Least-1}}

\newcommand{\LBM}{LBM}
\newcommand{\UBM}{UBM}
\newcommand{\WIM}{WIM}
\newcommand{\SIM}{SIM}
\newcommand{\SAM}{SAM}
\newcommand{\set}[1]{\{#1\}}

\newtheorem*{rtheorem}{\theremindertheorem}

\newcommand{\theremindertheorem}{}

\newenvironment{reminder}[1]
  {\renewcommand{\theremindertheorem}{Reminder of Theorem #1}\begin{rtheorem}}
  {\end{rtheorem}}

\newcommand{\BibTeX}{B\kern-.05em{\sc i\kern-.025em b}\kern-.08em\TeX}

\begin{document}


\begin{frontmatter}


\paperid{1149} 


\title{The Complexity of Manipulation of k-Coalitional Games on Graphs}


\author[A]{\fnms{Hodaya}~\snm{Barr} \footnote{odayaben@gmail.com}}
\author[A]{\fnms{Yohai}~\snm{Trabelsi}\footnote{yohai.trabelsi@gmail.com	}}
\author[A]{\fnms{Sarit}~\snm{Kraus}\footnote{sarit@cs.biu.ac.il	}}
\author[A]{\fnms{Liam}~\snm{Roditty}\footnote{Liam.Roditty@biu.ac.il}}
\author[B]{\fnms{Noam}~\snm{Hazon}\footnote{noamh@ariel.ac.il}}
\address[A]{Bar Ilan University}
\address[B]{Ariel University}

\begin{abstract}
In many settings, there is an organizer who would like to divide a set of agents into $k$ coalitions, and cares about the friendships within each coalition.
Specifically, the organizer might want to maximize utilitarian social welfare, maximize egalitarian social welfare, 
or simply guarantee that every agent will have at least one friend within his coalition. 
However, in many situations, the organizer is not familiar with the friendship connections, and he needs to obtain them from the agents. 
In this setting, a manipulative agent may falsely report friendship connections in order to increase his utility. In this paper, we analyze the complexity of finding manipulation in such $k$-coalitional games on graphs. 
We also introduce a new type of manipulation, socially-aware manipulation, in which the manipulator would like to increase his utility without decreasing the social welfare. We then study the complexity of finding socially-aware manipulation in our setting. Finally, we examine the frequency of socially-aware manipulation and the running time of our algorithms via simulation results.
\end{abstract}

\end{frontmatter}

\section{Introduction}
In many situations, there is an organizer who would like to divide a group of agents into $k$ non-empty coalitions.
For example, consider a manager who would like to divide his employees into $k$ teams in order to execute $k$ tasks. The interpersonal friendship connections between potential team members play an important role in such a setting.
Indeed, different managers may treat friendship differently. One manager may be interested in maximizing the number of friendship connections within all of the coalitions. Another manager may consider the welfare of an employee who is worse off than the others, and thus be interested in maximizing the minimum number of friendship connections that this employee has within his coalition.
It is also possible that a manager would simply require that every employee have at least one friend within his coalition.

The organizer may be familiar with all of the friendships among the agents. However, in some real-world scenarios, the organizer is unfamiliar with the friendships and thus needs to elicit them from the agents. For example, when dividing students into classes, it is common practice to ask them about their social relationships \cite{alonhighschool}.
In such situations, a manipulative agent, who is familiar with all the friendships among the agents, might have an incentive to misreport his friendship connections. Indeed, it is possible that such manipulation cannot be found efficiently.

In this paper, we analyze the complexity of finding manipulation in $k$-coalitional games \footnote{We refer to our model as ``$k$-coalitional game'' although we do not consider the notion of collusion and do not use, for example, the core.
This is based on previous papers that considered non-strategic agents in Hedonic games, focusing on agent division toward maximizing some social welfare functions (e.g., Aziz et al.~\cite{aziz2013computing}).}, 
with fixed $k$. 
We assume that the agents’ utilities depend on the friendship connections. Specifically, the friendship connections are represented by an unweighted graph, where the vertices are agents, and the edges represent the friendships among the agents. The utility of an agent is the number of friends he has within his coalition. There is an organizer who would like to divide the agents into exactly $k$ coalitions, but he builds the graph from the agents' reports. We analyze the settings where one manipulative agent would like to misreport his friendship connections by hiding some or reporting fake connections.

We study the objective of maximizing the egalitarian social welfare (\ME), and show that finding an optimal manipulation is computationally hard. Moreover, even deciding if a given report is beneficial for the manipulator is a hard problem. We then study a less demanding objective, in which the organizer requires that every agent has at least one friend within his coalition (\ALO). Indeed, finding an optimal manipulation or deciding if a given report is beneficial are still computationally hard problems. Arguably, the most natural organizer's objective is to maximize the utilitarian social welfare (\MU), and deciding if a given report is beneficial in this setting can be done efficiently. The complexity of finding an optimal manipulation with \MU\ remains open, but we provide an XP algorithm for this problem.

In addition, we introduce a new type of manipulation for coalitional games, socially-aware manipulation (\SAM), in which the manipulator would like to increase his utility without decreasing the social welfare. This manipulation models social situations in which a manipulator interested in his own welfare will not want to harm the welfare of society.
For example, consider an employee in a company or a player in a sports team. In these settings, the manipulator would like to maximize his utility without decreasing the social welfare, since it reduces the overall productivity or the teams’ performance.
Before analyzing the complexity of finding an \SAM, we show that for every objective (\MU, \ME, and \ALO), there are scenarios in which \SAM\ is possible. We then show that finding an optimal \SAM\ or deciding if a given report is an \SAM\ are still computationally hard problems with \ME\ or \ALO. Indeed, both problems can be solved in polynomial time with \MU. Finally, we provide simulation results based on a real social network. The results show that \SAM\ is quite frequent and demonstrate the effectiveness of our XP algorithm.


\section{Preliminaries}

Let $G=(A,E)$ be a directed graph representing a social network, where the vertex set $A$ represents a set of agents and the edge set $E$ represents friendship connections between the agents.
The graph is directed since friendship connections are not necessarily symmetric, and an edge $(a,a') \in E$ represents that $a$ considers $a'$ as his friend.
We denote by $E(a)$ the set of edges from $a$, and the set of neighbors of $a\in A$ by $N(a) = \set{a' \in A | (a,a') \in E(a)}$. 
For a subset $S\subseteq A$, we denote the set of neighbors of $a$ from the subset $S$ by $N(a, S) = N(a)\cap S$.
In our setting, there is an organizer who would like to divide the agents into exactly $k$ coalitions. Formally, the organizer seeks a partition $\mathcal{P}$ of $G$, which is a partition of the set $A$ into $k$ disjoint and non-empty sets $C_1,C_2, \ldots ,C_k$; we refer to these sets as coalitions.
Since $k$ is small in many settings, we further assume that $k$ is fixed.
Let $\Pi_k$ be the set of all partitions of size $k$.
We denote by $u(a, \mathcal{P})$ the number of friends that agent $a$ has within his coalition in partition $\mathcal{P}$, i.e, if $a\in C_i$ then $u(a,\mathcal{P})$ is $|N(a, C_i)|$.
When the organizer partitions the agents into coalitions, he may want to maximize a certain objective. In this paper, we examine three types of objectives:

\begin{definition}[\ME] 
The organizer wants to find a partition $\mathcal{P^*}$ that maximizes the minimum number of connections that an agent has within his coalition. That is, 
$\mathcal{P^*}=\arg \max_{\mathcal{P}\in \Pi_k} \min_{a\in A}u(a,\mathcal{P})$.
\end{definition}
\begin{definition}[\ALO] 
The organizer wants to find a partition $\mathcal{P^*}$ such that every agent has at least one connection within his coalition.
That is, $\mathcal{P^*}$ such that $\forall_{a\in A} u(a,\mathcal{P^*})>0$.
\end{definition}
\begin{definition}[\MU]
The organizer wants to find a partition $\mathcal{P^*}$ that maximizes the total number of connections within the same coalition. That is, $\mathcal{P^*} = \arg \max_{\mathcal{P}\in \Pi_k} \Sigma_{a\in A}u(a,\mathcal{P})$.
\end{definition}

Let  $O_{obj}(G)$ be the set of all partitions that satisfy the objective of the organizer. For $a\in A$, let $LB_{obj}(G,a)= \min_{\mathcal{P}\in O_{obj}(G)}u(a,\mathcal{P})$, and let $UB_{obj}(G,a)= \max_{\mathcal{P}\in O_{obj}(G)}u(a,\mathcal{P})$. That is, $LB_{obj}(G,a)$ and $UB_{obj}(G,a)$ are lower and upper bounds (respectively) on the number of friends that an agent $a$ can have in a partition satisfying the objective $obj$.

Note that when the objective of the organizer is \ALO\ it is possible that there is no feasible partition in $G$, i.e., $O_{\ALO}(G) = \emptyset$. In such a case, we define the utility of all the agents to be $0$. 

Each agent $a\in A$ reports to the organizer a set of friendship connections, $E^R(a)$, which is not necessarily equal to $E(a)$. That is, the organizer learns about the graph structure solely from the reports of the agents, and if all of the agents are truthful then this graph is equal to $G$.
\footnote{Since friendship connections are not necessarily symmetric, we assume that $G$ is a directed graph. If we assume that friendship connections are always symmetric, then $G$ is an undirected graph, which requires further assumptions on how the organizer builds the graph. We provide the definitions and results for undirected graphs in the appendix.}
%
%
A manipulator agent $m\in A$ reports a set of friendship connections $E^R(m)\neq E(m)$, 
so that the organizer will choose a partition that is better for $m$ than the partition that would have been chosen with $m$'s truthful report. In this case, we refer to $E^R(m)$ as the manipulation of $m$.
We assume that there is a single manipulator, $m$, that has full information regarding the other agents' reports, and the organizer's objective.
We examine two different types of manipulators, $m^+$ and $m^-$. 
\begin{definition} [$m^+$]
 A manipulator who can only add edges, $m^+ \in A$, reports a set of connections $E^R(m^+)=E(m^+)\cup E^+(m^+)$, where $E^+(m^+) \subseteq \{m^+\}\times A$. 
\end{definition}
\begin{definition}[$m^-$]
 A manipulator who can only remove edges, $m^-\in A$, reports a set of connections $E^R(m^-)=E(m^-)\setminus E^-(m^-)$, where  $E^-(m^-) \subset E(m^-)$ is the set of connections that $m$ does not  report to the organizer.
\end{definition}
Note that since $G$ is directed, $m$ is able to add and remove only outgoing edges. 
Let $G(m)=(A,(E \cup E^+(m)) \setminus E^-(m))$ be the graph after the manipulation of $m$.
Note that $m$ may have more than one possible manipulation, e.g., $E^R_1(m)$ and $E^R_2(m)$, which result in different graphs, $G_1(m)$ and $G_2(m)$, respectively.

We emphasize that $u(a,\mathcal{P})$ depends on the real graph $G$ (and not on $G(m))$.
That is, the utility of every agent is computed on the real graph. Therefore, if the manipulator adds an edge, this edge is not considered in $u(a,\mathcal{P})$. Similarly, if the manipulator removes an edge, this edge is still considered in $u(a,\mathcal{P})$. The manipulator influences only the beliefs of the organizer.


Given the organizer's objective, there may be several partitions that satisfy the objective, but the utility of the manipulator might be different in each such partition. 
We thus study four types of manipulations, following  \cite{waxman2021manipulation}:
\begin{definition}[Lower Bound Manipulation (\LBM)]
A manipulation $E^R(m)$ is \emph{\LBM} if $LB_{obj}(G(m),m) > LB_{obj}(G,m)$.
The improvement of an \LBM\ is $LB_{obj}(G(m),m) - LB_{obj}(G,m)$.
\end{definition}

\begin{definition}[Upper Bound Manipulation (\UBM)]
A manipulation $E^R(m)$ is \emph{\UBM} if $UB_{obj}(G(m),m) > UB_{obj}(G,m)$.
The improvement of a \UBM\ is $UB_{obj}(G(m),m) - UB_{obj}(G,m)$.
\end{definition}

That is, the goal of an \LBM\ is to eliminate partitions with low utility ($u(m,\mathcal{P})$), while the goal of an \UBM\ is to add partitions with high utility.

\begin{definition}[Weak-Improvement Manipulation (\WIM)]
A manipulation $E^R(m)$ is \emph{\WIM} if $LB_{obj}(G(m),m) > LB_{obj}(G,m)$ and $UB_{obj}(G(m),m) > UB_{obj}(G,m)$.
The improvement of a \WIM\  is $(LB_{obj}(G(m),m) + UB_{obj}(G(m),m)) - (LB_{obj}(G,m) + UB_{obj}(G,m))$.
\end{definition}

\begin{definition}[Strict-Improvement Manipulation (\SIM)]
A manipulation $E^R(m)$ is \SIM\ if $LB_{obj}(G(m),m) > UB_{obj}(G,m)$.
The improvement of an \SIM\ is $(LB_{obj}(G(m),m) + UB_{obj}(G(m),m)) - (LB_{obj}(G,m) + UB_{obj}(G,m))$.
\end{definition}

Let $type$ be one of the manipulation types (\LBM, \UBM, \WIM\ or \SIM). Given $E^R(m)$ and $type$, we denote by $I_{type}(E^R(m))$ the improvement of $E^R(m)$; if $E^R(m)$ is not a manipulation of the specific type, then $I_{type}(E^R(m)) = 0$.
Given $type$, we say that $E^R_1(m)$ is better than a $E^R_2(m)$ if $I_{type}(E^R_1(m)) > I_{type}(E^R_2(m))$, and an optimal manipulation is a manipulation with the maximum improvement.
See Figure \ref{fig:manipulation_types} for an illustration of the four types of manipulations.

Note that in \WIM\ and \SIM\ both the the lower bound and the upper bound increase. Therefore, we define the improvement for both these types as the sum of the improvements of the lower bound and the upper bound for both types.

\begin{figure}[H]

    \centering
    \includegraphics[width=1\columnwidth]{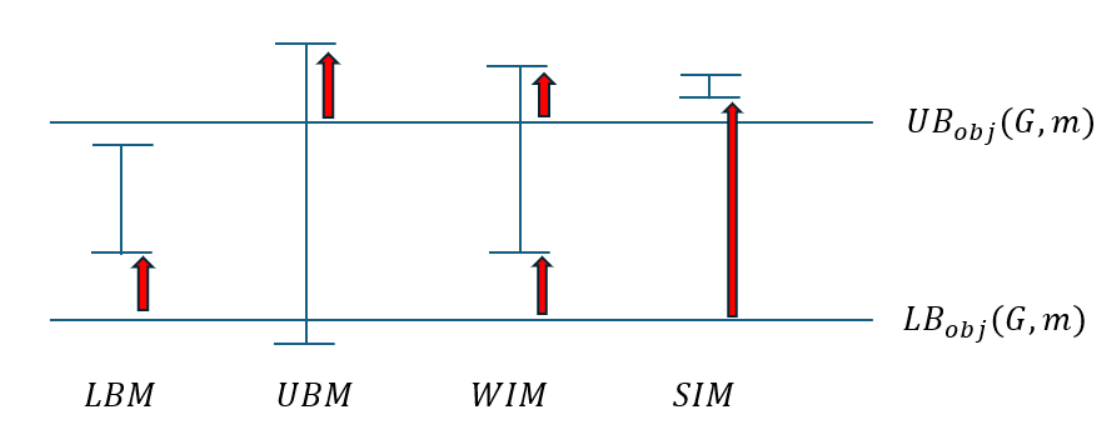}
    \caption{An illustration of the four types of           manipulations.\label{fig:manipulation_types}}
    \vspace{11pt}
\end{figure}

We say that an objective is \emph{susceptible} to \LBM\ by adding edges if there exists at least one graph in which there exists $m^+$ with an \LBM. Susceptibility by removing edges and to the other manipulation types is defined similarly. Whenever an objective is susceptible to manipulation, our main goal is to find an optimal manipulation (if the given instance is manipulable). Formally,
\begin{definition}[The Manipulation Problem]
    Given a graph $G$, an objective $obj$, a manipulator $m \in A$, and a manipulation type $type$, we are asked to find an optimal manipulation $E^R(m)$ (if such manipulation exists). 
\end{definition}

We also study an even simpler problem, which is to compute the improvement of a given report.
\begin{definition}[The Improvement Problem]
    Given a graph $G$, an objective $obj$, a manipulator $m \in A$, a manipulation type $type$, and a report $E^R(m)$, we are asked to compute $I_{type}(E^R(m))$.
\end{definition}




\section{Background} 
Our setting, in which the utility of an agent is the number of neighbors that he has within his coalition to which he is assigned \cite{sless2018forming}, is a special case of Additively Separable Hedonic Games (ASHGs), which have been extensively studied \cite{bogomolnaia2002stability,hanaka2022hedonic,aziz2016computational,aziz2011stable}. However, very few papers consider the problem of manipulation in ASHGs.
%
%
Indeed, 
Dimitrov and Sung~\cite{dimitrov2004enemies} analyzed ASHGs where agents have both positive and negative edges, and  provided a strategyproof algorithm for finding stable outcomes.
Rodr{\'\i}guez-{\'A}lvarez~\cite{rodriguez2009strategy} discussed strategyproof core stable solutions' properties. They proved that single lapping rules are necessary and sufficient for the existence of a unique core-stable partition.
 Aziz et al.~\cite{aziz2013pareto} proved that, with appropriate restrictions over the agents' preferences, the serial dictatorship mechanism is strategyproof. 
Flammini et al.~\cite{flammini2021strategyproof} studied the utilitarian social welfare in ASHGs and fractional hedonic games, and they proposed strategyproof mechanisms at the cost of non-optimal social welfare. 
They extended their analysis to friends and enemies games in a subsequent work~\cite{FLAMMINI2022103610}.
Wright and Vorobeychik~\cite{wright2015mechanism} considered a model of ASHG that is very similar to ours, but they restricted the size of each coalition instead of restricting the number of coalitions. In their work, they proposed a strategyproof mechanism that achieves good and fair experimental performance, but with no theoretical guarantee.
All of these works focused on developing strategyproof mechanisms,
while we study the computational complexity of finding manipulation.

The work that is closest to ours is by Waxman et al.~\cite{waxman2021manipulation}. 
They provided an extensive set of results specifying for each objective whether or not it is susceptible to manipulation.
Specifically, they show that \ME\ is susceptible to \SIM\ by removing edges (and thus, obviously, it is also susceptible to \LBM, \UBM, and \WIM), but it is not susceptible to any manipulation by adding edges.
\ALO\ is only susceptible to \LBM\ by removing edges (and not to \UBM, \WIM, or \SIM), and it is not susceptible to any manipulation by adding edges. 
\MU\ is susceptible to \SIM\ by adding or removing edges.
Indeed, \cite{waxman2021manipulation} did not study the computational complexity of finding manipulation. 

Alon~\cite{alonhighschool} considered the \ALO\ objective, and studied a different type of manipulation: whether a set of manipulators can guarantee to be in the same coalition. Alon showed that such manipulation is almost always impossible.

In our paper, we introduce a new type of manipulation, which we refer to as a socially-aware manipulation (Section \ref{sam_section}). In this manipulation, the manipulator $m$ would like to increase his utility without decreasing the social welfare. The tension between maximizing one's own utility and social welfare has been studied extensively in the social sciences (e.g., \cite{fehr2002social,lang2018explaining}). 
However, to the best of our knowledge, all of the works on manipulation in ASHGs assume that the manipulator aims to maximize his individual utility and does not care about the organizational social welfare.  

Table \ref{tab:mu_complexity} summarizes our complexity results with the \MU\ objective.
 With the \ME\ and \ALO\ objectives, both Manipulation Problem and Improvement Problem are computationally hard, for all settings susceptible to manipulation. Table \ref{tab:susceptible} summarizes the susceptible results from \cite{waxman2021manipulation}.

\begin{table}[]
    \centering
    \begin{tabular}{c|c|c}
         & Manipulation & Socially-Aware Manipulation \\
        Improvement Problem & P (T\ref{thm:mu_specific_man_improvment}) & P (T\ref{thm:sam_specific_man_improvment})\\
        Manipulation Problem & ? XP (T\ref{thm:alg_mu_regular}) & P (T\ref{thorem:mu_alg_not_harm})\\
    \end{tabular}
    \caption{Complexity results for \MU}
    \label{tab:mu_complexity}
\end{table}

\begin{table}[]
    \centering
    \begin{tabular}{c|c|c}
         & $Add$ & $Remove$\\
         \textbf{Max-Util} & \SIM & \SIM \\
         \textbf{At-Least-1} & No M & \LBM, No \UBM \\
         \textbf{Max-Egal} & No M& \SIM 
    \end{tabular}
    \caption{Summary of susceptible results from \cite{waxman2021manipulation}.}
    \label{tab:susceptible}
\end{table}

\section{The Complexity of the Manipulation and the Improvement Problems}

We begin with the \ME\ objective. Recall that manipulation is possible only by removing edges. Indeed, finding any type of optimal manipulation is computationally hard. 
Moreover, we show that even deciding whether a SIM exists for a given instance is computationally hard.
\begin{theorem}
Given a graph $G$, and a manipulator $m^- \in A$, deciding whether any type of manipulation exists when the objective is \ME\ is co-$NP$-hard.  
\label{thm:sim_me_directed_hard}
\end{theorem}

The reduction is from the complementary problem of the NP-complete $3$-SAT problem.
\begin{definition}[3-SAT problem]
Let $\mathcal{F}$ be a Boolean CNF formula, such that each clause has three literals. We are asked whether there exists a truth assignment that satisfies  $\mathcal{F}$.
\end{definition}

Our gadget for the hardness proof is a ring graph, which represents a CNF Boolean formula, as was introduced by \cite{bang2016finding}.

\begin{figure}[H]
          
       \hfill

        \centering
\includegraphics[page=1,width=0.1\textwidth]{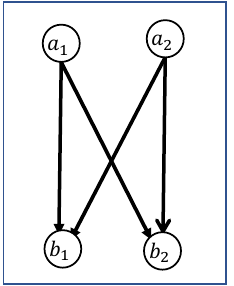}

    \caption{A switch } \label{fig:switch}
    
\end{figure}
\vspace{0.2cm}
\begin{figure}[hbp]
    
    \begin{minipage}{\linewidth}
          
        \hfill
         
            \begin{subfigure}{0.5\textwidth}
            \centering
            \includegraphics[page=1,width=\textwidth]{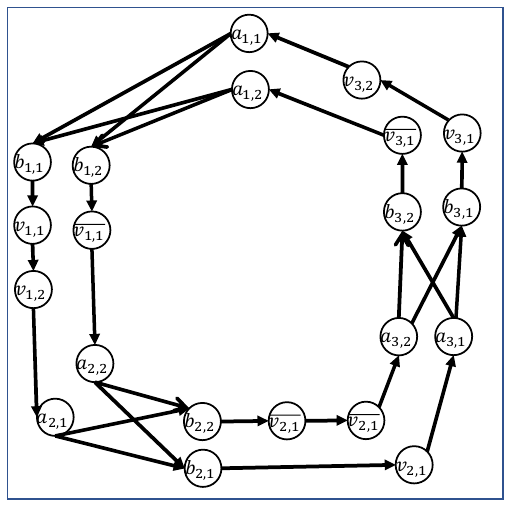}
 \caption{}
            \label{fig:ring}
        \end{subfigure}    
        \hfill
        \begin{subfigure}{0.5\textwidth}
            \centering
            \includegraphics[page=1,width=\textwidth]{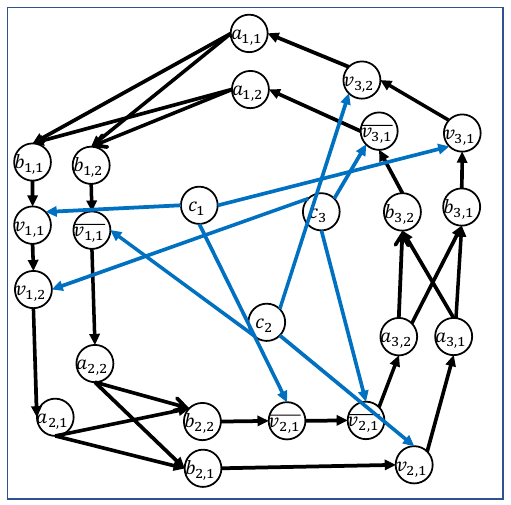}
            
\caption{}
            \label{fig:ring_clauses}
        \end{subfigure}

        \vspace{0.2cm}
        
    \caption{(a) An illustration of a ring graph for $\mathcal{F} = (x_1 \lor \bar{x_2} \lor x_3) \land (\bar{x_1} \lor x_2 \lor x_3) \land (x_1 \lor \bar{x_2} \lor \bar{x_3})$ (b) The ring graph for the same $\mathcal{F}$ with the additions of clause vertices.}
    \end{minipage}
\end{figure}
\begin{definition}[Ring graph]
Let  $\mathcal{F}$ be a  Boolean CNF formula with $m$ clauses, $C_1,\ldots, C_m$, and $n$ variables, $x_1,\ldots, x_n$.
A switch $s$ is a directed $K_{2,2}$ graph, that is, $s=(\{a_1,a_2,b_1,b_2\},\{a_1,a_2\}\times \{b_1,b_2\})$ (see Figure \ref{fig:switch}). 
The ring graph is composed  of $n$ switches, $s_1,\ldots, s_n$, where   $s_{i}$ is connected to $s_{i+1}$ with two directed paths (in the case that $i=n$ we connect $s_n$ to $s_1$), one  from $b_{i,1}$ to $a_{i+1,1}$ (the positive path) and another from $b_{i,2}$ to $a_{i+1,2}$ (the negative  path). For each positive (negative) occurrence of $x_i$ in $\mathcal{F}$ we add a vertex to the positive (negative) path. We refer to these vertices as literal vertices (see Figure \ref{fig:ring}).
\end{definition}

In addition to the basic ring graph structure, given a CNF formula we add $m$ clause vertices, $c_1, \ldots, c_m$. For each positive (negative) occurrence of $x_i$ in $C_j$ we add an edge from $c_j$ to one of the literal vertices on the positive (negative) path between $s_j$ and $s_{j+1}$, that has no incoming edge from another clause vertex (see Figure \ref{fig:ring_clauses}).

Note that a ring graph can be partitioned into exactly two disjoint cycles, where each cycle defines a truth assignment. This is due to the fact that in the switch, there is an option to choose between the positive path ($b_{i,1}$ to $a_{i+1,1}$) and the negative path ($b_{i,2}$ to $a_{i+1,2}$). There is only one directed path from $b_{i,1}$ to $a_{i+1,1}$ and only one directed path from $b_{i,2}$ to $a_{i+1,2}$. Therefore, a cycle that uses the positive path from  $b_{i,1}$ to $a_{i+1,1}$ represents a truth assignment in which $x_i$ is true, and a  cycle that uses the negative path from $b_{i,2}$ to $a_{i+1,2}$ represents a truth assignment in which $x_i$ is false.

Our proof works as follows.
Given a Boolean CNF formula $\mathcal{F}$, we construct a graph $G$ with a vertex $m^-$ that we identify as the manipulator, such that $m^-$ has a manipulation if and only if  $\mathcal{F}$ is not satisfiable.
In particular, we show that if a satisfying truth assignment exists, then there exists a graph partition where each vertex has at least $8$ neighbors (i.e., outgoing edges within his coalition). Since the manipulator has $8$ neighbors in $G$ there is not any type of manipulation.
On the other hand, if there is no satisfying truth assignment, we show that any optimal partition (according to the \ME\ objective) guarantees that each vertex gets at least $7$ neighbors in his coalition, and the manipulator $m^-$ gets exactly $7$ neighbors in his coalition. However, the manipulator has a \SIM, denoted as $E^-(m^-)$, such that $LB(G(m^-),m^-) = 8$. As \SIM\ makes the lower bound after the manipulation greater than the upper bound before the manipulation, every \SIM\ is also \LBM, \UBM, and \WIM. Therefore, the proof holds for all manipulation types.
Due to space constraints, the full proof and several other proofs are deferred to 
the appendix.

We now show that the improvement problem is also computationally hard. Moreover, even deciding whether a given report is a manipulation is computationally hard. The proof is almost identical to the proof of Theorem~\ref{thm:sim_me_directed_hard} 
\begin{theorem}
Given a graph $G$, a manipulator $m^- \in A$, a manipulation type $type$, and a report $E^R(m)$, deciding whether $I_{type}(E^R(m)) > 0$, when the objective is \ME\ is co-$NP$-hard.  
\label{thm:improvment_me_directed_hard}
\end{theorem}

The \ALO\ objective is less demanding than the \ME\ objective. 
In addition, it is susceptible only to \LBM\ by removing edges.
However, finding an optimal \LBM\ is computationally hard.
\begin{theorem}
Unless $P=NP$, there is no polynomial time algorithm for finding an optimal \LBM\ for a manipulator $m^-$ when the objective is \ALO.
\label{thm:alo_hard_compute}
\end{theorem}
The reduction is from a variant of the 3-SAT problem, in which each variable appears in at most $3$ clauses, each literal appears in at most $2$ clauses, and each clause has $2$ or $3$ literals.
The improvement problem is also computationally hard.
\begin{theorem}
Given a graph $G$, a manipulator $m^- \in A$,  and a report $E^R(m)$, deciding whether $I_{\LBM}(E^R(m)) > 0$  when the objective is \ALO\ is $NP$-hard.  
\label{thm:improvment_alo_directed_hard}
\end{theorem}
The reduction is from the same variant of 3-SAT that is used in the proof of Theorem~\ref{thm:alo_hard_compute}.

The \MU\ objective is susceptible to all types of manipulations, with either $m^+$ or $m^-$.
For the analysis of the \MU\ objective, we use the following definitions.

Given a directed graph $G=(A,E)$, we can naively convert $G$ to a weighted undirected graph $G^U=(A, E^U)$ by creating an edge $(u,v)$ whenever either $(u,v)$ or $(v,u)$ exists in $G$. If both $(u,v)$ and $(v,u)$ exist in $G$, then the edge weight is set to $2$. Otherwise, the edge weight is set to $1$. 
Given an undirected graph $G^U=(A,E^U)$, a set of edges $Y\subseteq E^U$ is a $k$-cut of $G^U$ if $G'=(A, E^U\setminus Y)$ is a graph that contains at least $k$ connected components. We refer to $2$-$cut$ simply as $cut$.
A $k$-cut $Y^*$ is a $min$-$k$-$cut$ if $\sum_{y \in Y^*} w(y)$ is minimal amongst all possible $k$-cuts of $G^U$. We denote $Y^*$ by $min$-$k$-$cut(G)$.
%
Note that any partition has a corresponding $k$-$cut$, and a \MU\ partition of $G$ has a corresponding $min$-$k$-$cut$  in $G^U$ (Based on Lemma 1 from ~\cite{branzei2009coalitional_and_stability}).
For simplicity, we will refer to the cuts in $G^U$ as cuts of $G$.

Unlike with the \ME\ and \ALO\ objectives, we show a polynomial time algorithm for solving the improvement problem with the \MU\ objective. 
\begin{theorem}
    The improvement problem can be solved in polynomial time when the objective is \MU, with any type of manipulator, and any manipulation type.
    \label{thm:mu_specific_man_improvment}
\end{theorem}
\begin{proof}
    The number of $min$-$k$-$cut$s of any graph is $\mathcal{O}(n^{2(k-1)})$ ~\cite{karger1996new}, and the set of all the $min$-$k$-$cut$s can be computed in polynomial time for a fixed $k$ (e.g., using the algorithm of~\cite{gupta2019number}). Since we assume that $k$ is fixed, the set $O_{\MU}(G)$, which contains the corresponding partitions for every $min$-$k$-$cut$ of $G$, can be computed in polynomial time.
    Similarly, the set $O_{\MU}(G(m))$, which contains the corresponding partitions for every $min$-$k$-$cut$ of $G(m)$, can be computed in polynomial time. Clearly, given $O_{\MU}(G)$ and $O_{\MU}(G(m))$, we can compute $LB(G,m)$, $LB(G(m),m)$, $UB(G,m)$, and $UB(G(m),m)$. Therefore, for any manipulation type, the improvement problem can be solved in polynomial time. 
\end{proof} 
As a consequence, given two reports, $E^R_1(m)$ and $E^R_2(m)$, we can decide whether $E^R_1(m)$ is better than $E^R_2(m)$ in polynomial time.
We believe that finding an optimal manipulation with the \MU\ objective is computationally hard, and we leave it as an important open problem. Indeed, we present a general XP algorithm for finding any type of optimal manipulation, with any type of manipulator. 
The algorithm is parameterized by $MMC$, the maximum size of the $min$-$k$-$cut$ that can result from any possible manipulation. Formally, $MMC = \max_{E^R(m)} |min$-$k$-$cut(G(m))|$. 


We first define a connection between a partition and a report of a manipulator.
\begin{definition}
 Let $\mathcal{P}=\{C_1,C_2\ldots C_k\}$ be a partition and let $m\in C_1$. 
When $m = m^+$,  we say that a report $E^R(m^+)$ respects $\mathcal{P}$ if  $E^+(m^+) \subseteq \{\{m^+\}\times C_1\}$ and  $\{\{m^+\}\times C_1\} \subseteq E^R(m^+)$.
When $m = m^-$,  we say that a report $E^R(m^-)$ respects $\mathcal{P}$ if $E^-(m^-) \subseteq \{\{m^-\}\times C_i| i \in [2,k]\}$ and $E^R(m^-) \cap \set{\{m^-\}\times C_i|  i \in [2,k]} = \emptyset$.
\end{definition}
That is, a report $E^R(m^+)$ respects a partition $\mathcal{P}$ if the manipulator adds only edges to vertices within his coalition ($C_1$), and $G(m^+)$ contains all the edges between $m^+$ and the vertices within his coalition. Similarly, a report $E^R(m^-)$ respects a partition $\mathcal{P}$ if the manipulator removes only edges to vertices outside of his coalition, and $G(m^-)$ does not contain any edge between $m^-$ and the vertices outside of his coalition. Note that a report $E^R(m)$ that respects a partition $\mathcal{P}$ is not necessarily better for $m$ than the truthful report $E(m)$. If $E^R(m)$  respects a partition $\mathcal{P}$ and $E^R(m)$  is better for $m$ than $E(m)$, we say that $E^R(m)$ is a manipulation that respects a partition $\mathcal{P}$.

We now define a specific set of partitions, $\Pbar$.
$\Pbar$ contains every partition $\mathcal{P}$, which is surely beneficial for $m$, that the organizer may select as a result of a manipulation of $m$.
That is, there is no manipulation by $m$ such that $\mathcal{P} \in O_{\MU}(G(m))$ and $u(m, \mathcal{P}) > LB_{\MU}(G,m)$ but $\mathcal{P} \notin \bar{\mathcal{P}}$.

Our XP algorithm, Algorithm~\ref{alg:algorithm_max_util}, works as follows. It first calls Algorithm~\ref{alg:compute_p_bar_regular_manipulation} for computing $\bar{\mathcal{P}}$. It then iterates over all the partitions in $\Pbar$. For each $\mathcal{P} \in \Pbar$, the algorithm computes the report $E^R(m)$ that respects $\mathcal{P}$. The algorithm then compares $E^R(m)$ with the current best manipulation and updates it accordingly. 
The algorithm returns the best manipulation that respects a partition in $\Pbar$ (if such manipulation exists).

In order to prove the correctness of Algorithm~\ref{alg:algorithm_max_util}, we first show that Algorithm \ref{alg:compute_p_bar_regular_manipulation} correctly computes the set $\Pbar$.

\begin{algorithm}

\begin{algorithmic}[1] 
\REQUIRE $G=(A,E) , m\in A$\\
\ENSURE The set $\bar{\mathcal{P}}$

\STATE Let $E^R_{all}(m) = \{(m,a)|a\in A\}$
\IF {$m=m^+$}
\STATE {$maxSize \gets |min$-$k$-$cut(G_{all}(m))|$}
\ELSE [$\slash * m=m^- * \slash$]

\STATE {$maxSize \gets 2 \cdot |min$-$k$-$cut(G)|$}

\ENDIF
\FORALL{$k$-$cut$ $Y$ in $G$ of size at most $maxSize$}
    \STATE $\mathcal{P} \gets$ the corresponding partition of $Y$
    \IF {$u(m, \mathcal{P}) > LB_{\MU}(G,m)$}
        \STATE Add $\mathcal{P}$  to $\Pbar$
    \ENDIF
\ENDFOR
\RETURN  $\Pbar$
\end{algorithmic}
\caption{Compute the set $\bar{\mathcal{P}}$.}

\label{alg:compute_p_bar_regular_manipulation}
\end{algorithm}

\begin{algorithm}[bhpt]
\caption{Manipulation for \MU\ objective}
\label{alg:algorithm_max_util}

\begin{algorithmic}[1] 
\REQUIRE $G=(A,E) , m\in A$\\
\ENSURE $E^R(m)$

\STATE Compute the set $\bar{\mathcal{P}}$ by Algorithm \ref{alg:compute_p_bar_regular_manipulation} \label{line:compute_pbar}
\STATE $manip \gets E(m)$
\FORALL{$\mathcal{P}\in \bar{\mathcal{P}}$}

\STATE Compute $E^R(m)$ that respects $\mathcal{P}$
\IF {$E^R(m)$  is better then $manip$}
\STATE $manip \gets E^R(m)$
\ENDIF
\ENDFOR
\IF {$manip = E(m)$}
\RETURN No manipulation
\ENDIF
\RETURN $manip$
\end{algorithmic}
\end{algorithm}

\begin{lemma}
    \label{lemma:alg_correct_pbar}
    Algorithm \ref{alg:compute_p_bar_regular_manipulation} correctly computes the set $\bar{\mathcal{P}}$.
\end{lemma}
\begin{proof}
    Assume by contradiction that there is a manipulation $E^R_1(m)$ such that $\mathcal{P}\in O_{\MU}(G_1(m))$ and $u(m, \mathcal{P}) > LB_{\MU}(G,m)$  but $\mathcal{P} \not\in \bar{\mathcal{P}}$. 
    Let $Y_1$ be $\mathcal{P}$'s corresponding $k$-$cut$ in $G_1(m)$.
    Since any manipulation does not change the number of vertices, a partition in $G_1(m)$ is also a partition in $G$.         
    Let $Y$ be $\mathcal{P}$'s corresponding $k$-$cut$ in $G$.
        
        \textbullet{\ If $m=m^+$:}
        Adding edges can only increase the size of the $min$-$k$-$cut$ of a graph. That is, $|min$-$k$-$cut(G)| \leq |min$-$k$-$cut(G_1(m))|$.
        Since $E^R_{all}(m)$ consists of all of the edges from $m$, $|min$-$k$-$cut(G_1(m))| \leq |min$-$k$-$cut(G_{all}(m))| = maxSize$.  Thus, $|Y| \leq |Y_1| = |min$-$k$-$cut(G_1(m))| \leq maxSize$.
        Therefore, by the algorithm's construction, $\mathcal{P} \in \Pbar$. A contradiction.
    
    \textbullet{\ If $m = m^-$:}
        Removing edges can only decrease the size of the $min$-$k$-$cut$ of a graph. Hence:  
        $|Y_1| \leq |min$-$k$-$cut(G)|$.
        In addition, the difference between the sizes of $Y$ and $Y_1$ is at most the number of edges that $m$ is able to remove. That is,  
        $|Y|  - |Y_1| \leq |Y \cap E(m)|$.
        Therefore, 
       \begin{equation}
            \label{eq:diff_Y}
            |Y| - |Y \cap E(m)|  \leq |min\text{-}k\text{-}cut(G)|
       \end{equation}
        Let $\mathcal{P}_{LB} \in O_{\MU}(G)$ be a partition such that $u(m, \mathcal{P}_{LB}) = LB_{\MU}(G,m)$, and let $Y_{LB}$ be the corresponding $k$-$cut$. Recall that $Y_{LB}$ is a $min$-$k$-$cut$ in $G$.
        Since $u(m, \mathcal{P}) > LB_{\MU}(G,m)$, then
        $|Y \cap E(m)| < |Y_{LB} \cap E(m)| \leq  |min$-$k$-$cut(G)|$.
        Combined with \eqref{eq:diff_Y} we get that $|Y| < 2 \cdot|min$-$k$-$cut(G)| = maxSize$.
        Therefore, by the algorithm's construction, $\mathcal{P} \in \Pbar$. A contradiction.
%
\end{proof}

We now show that if there is a manipulation $E^R_{*}(m)$, then there exists a partition $\mathcal{P} \in \bar{\mathcal{P}}$ such that the report that respects $\mathcal{P}$, $E^R(m)$, is a manipulation, and $E^R_{*}(m)$ is not better than $E^R(m)$. Our proof is for \LBM\ by a manipulator $m^-$. The proof for $m^+$ is very similar, and we provide it in the appendix. It is also straightforward to adapt both proofs for \UBM. However, for \WIM\ and \SIM\ the proofs are very similar, but they are slightly different from the proofs for \LBM\ and \UBM. We thus provide the proof for \WIM\ and $m^-$ in the appendix.

\begin{lemma}
If  $E_1^R(m^-)$ is an \LBM\ in which $LB_{\MU}(G_1(m^-), m^-) = x$, then there is a manipulation $E_2^R(m^-)$ such that (i) $E_2^R(m^-)$ respects a partition $\mathcal{P}\in \bar{\mathcal{P}}$  and  (ii) $LB_{\MU}(G_2(m^-), m^-) = x$.
\label{lemma:remove_regular manipulation_lb}
\end{lemma}
\begin{proof}
Given an \LBM\ $E_1^R(m^-)$ in which $LB_{\MU}(G_1(m^-), m^-) = x$, let $\mathcal{P}$ be a partition such that 
$\mathcal{P}\in O_{\MU}(G_1(m^-))$ and $u(m^-,\mathcal{P}) = x$. Let $\mathcal{P} =\{ C_1, C_2 \ldots C_k\}$ and $m^-\in C_1$. By definition of $\Pbar$, $\mathcal{P}\in \Pbar$.
Let $E_2^R(m^-)$ be the report that respects $\mathcal{P}$.
We need to show that $LB_{\MU}(G_2(m^-), m^-) = x$. 
We do so by showing that 
$O_{\MU}(G_2(m^-)) \subseteq O_{\MU}(G_1(m^-))$ and $\mathcal{P} \in O_{\MU}(G_2(m^-))$.

Let $Y$ be the $k$-$cut$ that corresponds to $\mathcal{P}$ in $G$, and let $Y_i$ be the $k$-$cut$ that corresponds to $\mathcal{P}$ in $G_i(m^-)$, where $i\in \{ 1,2\}$. 
%
Given $\mathcal{P}' \not\in O_{\MU}(G_1(m^-))$, let $Y'$ be its corresponding $k$-$cut$ in $G$, and let $Y'_i$ be its corresponding $k$-$cut$ in $G_i(m^-)$, where $i\in \{ 1,2\}$.
Given $\mathcal{P}''\in O_{\MU}(G_1(m^-))$, let $Y''$ be its corresponding $k$-$cut$ in $G$, and let $Y''_i$ be its corresponding $k$-$cut$ in $G_i(m^-)$, where $i\in \{ 1,2\}$.
Let $q=|Y_2|-|Y_1|$, let $q'=|Y'_2|-|Y'_1|$ and let $q''=|Y''_2|-|Y''_1|$.

We begin by proving the following claim:
\begin{Claim}
    $q\leq q'$  and  $q\leq q''$.
    \label{claim:q_q'}
\end{Claim}


\begin{proof}

We divide the set $ E_1^-(m^-)\cup E_2^-(m^-)$  into the following three disjoint sets:  
\begin{enumerate}
\item The set $H_1= E_1^-(m^-)\cap E_2^-(m^-)$. 
\item The set $H_2= E_1^-(m^-)\setminus H_1$.
\item The set $H_3= E_2^-(m^-)\setminus H_1$.
\end{enumerate}

The set $H_1$ is contained in both $E_1^-(m^-)$ and $E_2^-(m^-)$, and thus the edges from $H_1$ are in neither $G_1(m^-)$ nor $G_2(m^-)$. That is, the edges from $H_1$ are not in $Y_i$, $Y'_i$ and $Y''_i$, for $i \in \{1,2\}$.
Since  $q=|Y_2|-|Y_1|$, $q'=|Y'_2|-|Y'_1|$ and $q''=|Y''_2|-|Y''_1|$, then the values of $q$, $q'$ and $q''$ do not depend on the edges from $H_1$. 

Now consider the set $H_2$. The edges from $H_2$ are not in $G_1(m^-)$ but they are in $G_2(m^-)$.
Therefore, the edges of $H_2$ are not in $Y_1$, $Y'_1$, and $Y''_1$.
The manipulation  $E^R_2(m)$ respect $\mathcal{P}$, and thus $H_2\subseteq \{ (m^-,x)\mid x\in C_1\}$.
Therefore, the edges of $H_2$ are not included in any $k$-$cut$ that corresponds to $\mathcal{P}$, i.e., the edges of $H_2$ are also not in $Y_2$. That is, the value of $q$ does not depend on the edges from $H_2$.
However, the edges of $H_2$ may be in $Y'_2$ or in $Y''_2$, and each edge from $H_2$ that is in $Y'_2$ or in $Y''_2$ increases the value of $q'$ or $q''$, respectively. 

Finally, consider the set $H_3$. 
The edges from $H_3$ are not in $G_2(m^-)$ but they are in $G_1(m^-)$.
Therefore, the edges of $H_3$ are not in $Y_2$, $Y'_2$, and $Y''_2$.
The manipulation $E^R_2(m)$ respects $\mathcal{P}$, and thus $H_3\subseteq \{ (m^-,x)\mid x\in \{C_i| i \in [2,k]\}\}$.  Therefore, the edges of $H_3$ are in any $k$-$cut$ that corresponds to $\mathcal{P}$, i.e., the edges of $H_3$ are in $Y_1$. That is, each edge from $H_3$ decreases the value of $q$.
In addition, the edges of $H_3$ may be in $Y'_1$ or in $Y''_1$, and each edge from $H_3$ that is in $Y'_1$ or in $Y''_1$ decreases the value of $q'$ or $q''$, respectively. 

Overall, $q\leq q'$ and $q\leq q''$.
\end{proof}

We now show that $\mathcal{P} \in O_{\MU}(G_2(m^-))$, and we do so by showing that $Y_2$ is a $min$-$k$-$cut$ in $G_2(m^-)$.
$\mathcal{P}' \not\in O_{\MU}(G_1(m^-))$, and thus  $Y'_1$ is not a $min$-$k$-$cut$ in $G_1(m^-)$. On the other hand, $\mathcal{P} \in O_{\MU}(G_1(m^-))$, and thus $Y_1$ is a $min$-$k$-$cut$ in $G_1(m^-)$. That is,  $|Y_1| < |Y'_1|$. From Claim \ref{claim:q_q'} we have that $q\leq q'$, and thus $|Y_2| < |Y'_2|$. 
Now, $\mathcal{P}'' \in O_{\MU}(G_1(m^-))$ and $\mathcal{P} \in O_{\MU}(G_1(m^-))$. That is, both $Y_1$ and $Y''_1$ are $min$-$k$-$cut$s in $G_1(m^-)$. Therefore, $|Y_1| = |Y''_1|$. From Claim \ref{claim:q_q'} we have that $q\leq q''$ and thus $|Y_2| \leq |Y''_2|$. Overall, $|Y_2| < |Y'_2|$, and  $|Y_2| \leq |Y''_2|$. That is, the size of $Y_2$ is at most the size of any $k$-$cut$. Therefore, $|Y_2|$ is a $min$-$k$-$cut$ in $G_2(m^-)$.

It remains to show that $O_{\MU}(G_2(m^-)) \subseteq O_{\MU}(G_1(m^-))$. Indeed, we showed that  $|Y_2| < |Y'_2|$. Thus, $Y'_2$ is not a $min$-$k$-$cut$ in $G_2(m^-)$. Therefore, $\mathcal{P}' \not\in O_{\MU}(G_2(m^-))$.
\end{proof}

Overall, we get:
\begin{theorem}
    Algorithm \ref{alg:algorithm_max_util} solves the manipulation problem with the \MU\ objective in time $\mathcal{O}(n^{2(maxSize + k)})$.
    \label{thm:alg_mu_regular}
\end{theorem}
\begin{proof}
    The correctness of the algorithm is a direct consequence of the Lemmas \ref{lemma:alg_correct_pbar} and \ref{lemma:remove_regular manipulation_lb}. As for the running time, Algorithm~\ref{alg:algorithm_max_util} calls  Algorithm~\ref{alg:compute_p_bar_regular_manipulation} for computing $\Pbar$. Algorithm~\ref{alg:compute_p_bar_regular_manipulation} computes all the $k$-$cut$s of size at most $maxSize$, and checks if the corresponding partition is beneficial for $m$. There are at most ${(n^2)}^{maxSize}$ subsets of edges of size at most $maxSize$. In addition, computing the corresponding partition for each subset of edges and checking if the partition is beneficial for $m$ takes at most $O(n^2)$. Therefore, the running time of Algorithm~\ref{alg:compute_p_bar_regular_manipulation} is $\mathcal{O}(n^{2maxSize} \cdot n^2)$. After computing $\Pbar$, Algorithm~\ref{alg:algorithm_max_util} iterates over all $\mathcal{P} \in \Pbar$. In each iteration, the algorithm first computes $E^R(m)$ that respects $\mathcal{P}$, which takes at most $\mathcal{O}(n^2)$. It then checks if $E^R(m)$ is better than $manip$, and thus it needs to iterate over all the $min$-$k$-$cut$s of $G(m)$. Computing all the $min$-$k$-$cut$ takes at most $\mathcal{O}(n^{2k})$ \cite{gupta2019number}. 
    Overall, the running time is $\mathcal{O}(n^{2maxSize} \cdot n^2 + n^{2maxSize} \cdot(n^2 + n^{2k})) = O(n^{2(maxSize + k)})$.
\end{proof}

Recall that $MMC$ is the maximum size of the $min$-$k$-$cut$ that can result from any possible manipulation. 
Now, removing edges can only decrease the size of the $min$-$k$-$cut$ of a graph. Therefore, if $m=m^-$, $MMC = |min$-$k$-$cut(G)| = \frac{1}{2} maxSize$. Adding edges can only increase the size of the $min$-$k$-$cut$ of a graph. Therefore, if $m=m^+$, $MMC = |min$-$k$-$cut(G_{all}(m))| = maxSize$ (since $E^R_{all}(m)$ consists of all of the edges from $m$). 
Therefore, Theorem \ref{thm:alg_mu_regular} essentially shows that Algorithm \ref{alg:algorithm_max_util} is an XP algorithm, parameterized by $MMC$.


Note that for $m^-$, it is possible to slightly improve Algorithm~\ref{alg:algorithm_max_util}, to get a running time of $\mathcal{O}(n^{\frac{1}{2}maxSize + 2k)})$. The details are in the appendix.

\section{Socially-aware Manipulation}

\label{sam_section}
We now introduce a new type of manipulation, which we call Socially-aware Manipulation (\SAM). In \SAM, the manipulator $m$ would like to increase his utility without decreasing the social welfare. That is, $O_{obj}(G(m)) \subseteq O_{obj}(G)$, and thus SAM is necessarily an \LBM.
\begin{definition}[Socially-Aware Manipulation (\SAM)]
A manipulation $E^R(m)$ is SAM if $O_{obj}(G(m)) \subseteq O_{obj}(G)$ and $LB_{obj}(G(m),m) > LB_{obj}(G,m)$.
\end{definition}
Note that term SAM  may remind the concept of Pareto optimality.
However, in SAM, the manipulator aims to increase his utility without decreasing the social welfare. That is, the manipulator would not want to harm the organizer's objective, but his manipulation may decrease the utility of one of the agents, unlike the notion of Parteo optimality. 
\subsection{The Existence of \SAM}
Before analyzing the complexity of finding an \SAM, we establish the existence of \SAM.

\begin{theorem}
The susceptibility of \SAM\ is equivalent to the susceptibility of \LBM. Specifically, \MU\ is susceptible to \SAM\ by adding or removing edges, but \ME\ and \ALO\ are susceptible to \SAM\ only by removing edges.
\label{thm:exist-no-harm-man}
\end{theorem}
\begin{proof}

Consider the graph as depicted in Figure\ref{fig:Util_add}. Assume that $k=2$, the organizer's objective is \MU, and $m=m^+$. 
Note that $|min$-$cut(G)|=2$. In addition, $\mathcal{P}_1 = \{\{m\}, \{a,b,c,d\}\} \in O_{\MU}(G)$, and $u(m,\mathcal{P}_1)=0$. Thus, $LB_{\MU}(G,m)=0$. By adding the dashed edge (from $m$ to $a$), the size of $min$-$cut(G(m))$ is still $2$. That is, for each $\mathcal{P} \in O_{\MU}(G(m))$ the size of the corresponding $cut$ in $G(m)$ is $2$.
Since $m$ does not remove any edge, the size of $\mathcal{P}$'s corresponding $cut$ in $G$ is $2$, and thus $\mathcal{P} \in O_{\MU}(G)$. That is, $O_{\MU}(G(m)) \subseteq O_{\MU}(G)$. In addition, every $cut$ in $G(m)$ that contains both $(m,c)$ and $(m,d)$ is of size at least $3$. Therefore, every $min$-$cut$ in $G(m)$ contains at most one edge from $m$. That is, $LB_{\MU}(G(m),m) >0$, and the manipulation is \SAM.

Figure \ref{fig:Util_remove} shows a graph in which \SAM\ is possible with \MU\ and $m=m^-$.
Figure \ref{fig:Egal-remove} shows a graph in which \SAM\ is possible with \ME\ and $m=m^-$. The full proof is in the appendix.

 

For \ALO, we will show below that any \LBM\ is also \SAM\ (Theorem \ref{thm:ALO_m-_is_sam}). Therefore, since \ALO\ is susceptible to \LBM\ by removing edges, it is also susceptible to \SAM\ by removing edges.

Finally, it was already shown by \cite{waxman2021manipulation} that \ME\ and \ALO\ are not susceptible to \LBM\ by adding edges. Therefore, they are not susceptible to \SAM\ by adding edges (since each \SAM\ is an \LBM).
\end{proof}
Note that the equivalence between the susceptibility of \SAM\ and  \LBM\ (as shown in Theorem \ref{thm:exist-no-harm-man}) means that any \emph{setting} (i.e., objective and type of manipulator) that is susceptible to \LBM\ is also susceptible to \SAM.

\begin{figure}[hbpt]
        \begin{subfigure}{0.18\linewidth}
            \centering
            \includegraphics[page=1,width=\linewidth]{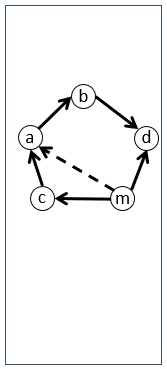}
            \caption{}
            \label{fig:Util_add}
        \end{subfigure}
        \hfill
        \begin{subfigure}{0.18\linewidth}
            \centering
            \includegraphics[page=1,width=\linewidth]{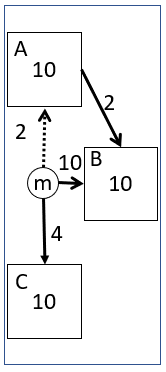}
            \caption{}
            \label{fig:Util_remove}
        \end{subfigure}
        \hfill
        \begin{subfigure}{0.18\linewidth}
            \centering
            \includegraphics[page=1,width=\linewidth]{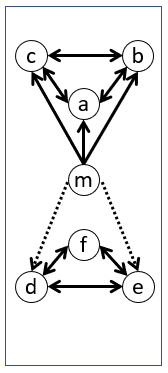}
            \caption{}
            \label{fig:Egal-remove}
        \end{subfigure}
        \vspace{10pt}
        \caption{Graphs for Theorem~\ref{thm:exist-no-harm-man}. A dashed line is an edge that is added by $m$. A dotted line is an edge that is removed by $m$. A square represents a clique of size $10$, and an edge with a number represents a set of edges.}
        
        \label{fig:exist-no-harm-graphs}
        
\end{figure}

\subsection{The Complexity of the Manipulation and Improvement Problems for SAM}
We begin with the \ME\ objective. 
Similar to \LBM, deciding whether an \SAM\ exists (and thus also finding an optimal \SAM) is computationally hard.
\begin{theorem}
Given a graph $G$, and a manipulator $m^- \in A$, deciding whether an \SAM\ exists when the objective is \ME\ is co-NP-hard. 
\label{thorem:sam_max_egal_np_directed}
\end{theorem}
The improvement problem is also hard.
\begin{theorem}
Given a graph $G$, a manipulator $m^- \in A$, and a report $E^R(m)$, deciding whether $I_{\SAM}(E^R(m))>0$ when the objective is \ME\ is co-NP-hard. 
\label{thorem:sam_max_egal_np_directed_improvment}
\end{theorem}

With the \ALO\ objective, we show that every \LBM\ is also an \SAM, and thus the manipulation and improvement problems for \SAM\ are computationally hard. 
\begin{theorem}
    For \ALO, every \LBM\ by removing edges is also \SAM\ by removing edges.
    \label{thm:ALO_m-_is_sam}
\end{theorem}
\begin{proof}
    Assume by contradiction that there exists a graph $G=(A,E)$, and a manipulator $m\in A$, such that exists an \LBM, $E^-(m)$ that is not \SAM. i.e., $LB(G(m), m) > LB(G,m)$ but $O_{\ALO}(G(m))\not \subseteq O_{\ALO}(G)$. 
    Let $\mathcal{P}$ be a partition such that $\mathcal{P} \in O_{\ALO}(G(m))$ but $\mathcal{P} \not\in O_{\ALO}(G)$. Therefore, there is a vertex $a\in A$ such that $u(a, \mathcal{P}) = 0$ in $G$. 
    However, $m$ can only remove outgoing edges, and thus $u(a, \mathcal{P}) = 0$ in $G(m)$ too. That is, $\mathcal{P} \not\in O_{\ALO}(G(m))$. A contradiction.
\end{proof}
Now, with the \MU\ objective, the improvement problem for \SAM\ can be solved in polynomial time. 

\begin{theorem}
    The improvement problem  for \SAM\ can be solved in polynomial time when the objective is \MU.
    \label{thm:sam_specific_man_improvment}
\end{theorem}
\begin{proof}
    As in the proof of Theorem \ref{thm:mu_specific_man_improvment}, $O_{\MU}(G)$ and $O_{\MU}(G(m))$ can be computed in polynomial time. We also need to check whether $O_{\MU}(G(m)) \subseteq O_{\MU}(G)$, which takes at most $O(|O_{\MU}(G(m))|\cdot|O_{\MU}(G)|\cdot n)$ steps.
\end{proof}

Interestingly, the manipulation problem for \SAM\ with the \MU\ objective can also be solved in polynomial time. Indeed, we only need to change line~\ref{line:compute_pbar} in Algorithm~\ref{alg:algorithm_max_util}, so that $\Pbar \gets O_{\MU}(G)$. 
To prove the correctness, we first show an essential property of every report that respects a \MU\ partition. 
\begin{lemma}
If $E^R(m)$ is a report that respects $\mathcal{P}\in O_{\MU}(G)$  then $O_{\MU}(G(m)) \subseteq O_{\MU}(G)$.
\label{lemma:not_harm}
\end{lemma}
\begin{proof}
Let $\mathcal{P} \in O_{\MU}(G)$, and let $E^R(m)$ be the report that respects $\mathcal{P}$. 
Let $Y$ be the $min$-$k$-$cut$ that corresponds to $\mathcal{P}$ in $G$. 
Given $\mathcal{P}'\not\in O_{\MU}(G)$  let $Y'$ be its corresponding $k$-$cut$ in $G$.
Since  $Y$ is a $min$-$k$-$cut$ and $Y'$ is not a $min$-$k$-$cut$, we have that $|Y|<|Y'|$.
Let $Y_1$ be the $k$-$cut$ that corresponds to  $\mathcal{P}$ in $G(m)$ and  $Y'_1$ be the $k$-$cut$ that corresponds to  $\mathcal{P}'$ in $G(m)$.

\textbullet{\ If $m = m^-$}, since $E^R(m)$ respects $\mathcal{P}$,  $E^-(m) \subseteq  Y$.
Therefore,  $|E^-(m) \cap  Y| = |E^-(m)| \geq  |Y'\cap E^-(m)|$.
Thus, $|Y_1| = |Y|-|E^-(m) \cap  Y| <  |Y'| -  |Y'\cap E^-(m)| =  |Y'_1|$.
That is, $\mathcal{P}'\not\in O_{\MU}(G(m))$.

\textbullet{\ If  $m = m^+$}, since $E^R(m)$ respects $\mathcal{P}$,  $E^+(m) \cap Y_1 = \emptyset$.  
Therefore,
$|Y_1| = |Y|+|E^+(m) \cap  Y_1| = |Y| < |Y'| \leq |Y'| + |E^+(m) \cap  Y'_1| = |Y'_1|$.
That is, $\mathcal{P}'\not\in O_{\MU}(G(m))$.
\end{proof}



We now show that if there is an \SAM\ $E^R_{*}(m)$, then there exists a partition $\mathcal{P} \in \bar{\mathcal{P}}$ such that the report that respects $\mathcal{P}$, $E^R(m)$, is an \SAM, and $E^R_{*}(m)$ is not better than $E^R(m)$. We provide the proof for $m^-$. The proof for $m^+$ is a straightforward adaptation. 

\begin{lemma}
If  $E_1^R(m^-)$ is an \SAM\ in which $LB_{\MU}(G_1(m^-), m^-) = x$, then there is an \LBM\ $E_2^R(m^-)$ such that (i) $E_2^R(m^-)$ respects a partition $\mathcal{P}\in \bar{\mathcal{P}}$ and (ii) $LB_{\MU}(G_2(m^-), m^-) = x$.
\label{lemma:remove_SAM}
\end{lemma}


\begin{proof}
We use the same definitions and claims as in the proof of Lemma \ref{lemma:remove_regular manipulation_lb}
Indeed, here $\Pbar$ is $O_{\MU}(G)$, but since $E_1^R(m^-)$ is \SAM\ then $\mathcal{P}$ is in $O_{\MU}(G)$, as required.
\end{proof}

\begin{theorem}
    Algorithm \ref{alg:algorithm_max_util}, in which line \ref{line:compute_pbar} is $\bar{\mathcal{P}} \gets O_{\MU}(G)$, solves the manipulation problem for \SAM\ with the \MU\ objective in time $\mathcal{O}(n^{2(2k-1)})$.
     \label{thorem:mu_alg_not_harm}
\end{theorem}
\begin{proof}
The correctness of the algorithm follows directly from Lemmas \ref{lemma:not_harm} and \ref{lemma:remove_SAM}.
As for the running time, we use the same analysis as in the proof of Theorem \ref{thm:alg_mu_regular}, but now $|\Pbar|$ is $\mathcal{O}(n^{2(k-1)})$. Thus, the running time is $\mathcal{O}(n^{2(2k-1)})$.      
\end{proof}

\section{Experiments}
In order to examine the frequency of \SAM\ and the effectiveness of our XP algorithm with the \MU\ objective, we ran some simulations \footnote{Our code is available in github: https://github.com/hodayaBen/The-Complexity-of-Manipulation-of-k-Coalitional-Games-on-Graphs.
In our implementation we compute all-min-cuts by using a sub-routine from VieCut that computes the Cactus of a graph \cite{ henzinger2020finding}.} on graphs that are based on the
Twitter followers dataset \cite{leskovec2012learning}. 
Since it was too large, we sampled $150$ subgraphs of the network.
Specifically, we sampled subgraphs of sizes $5,10,15,20$, and $25$. 
For each graph size, we first sampled at least $9$ graphs: we randomly chose a vertex and ran BFS where in each iteration we added only $2$ of the neighbors to the queue. The search was terminated when it reached the desired graph size. In order to vary the number of edges we repeated this process once when we added $4$ neighbors to the queue, and then when we added $6$ neighbors.
If we did not reach the desired graph size, we sampled again an initial vertex (up to $100$ times). 

We set $k=2$ and $m=m^+$. 
Note that we concentrated on manipulation by adding edges, since it was quite common; e.g., we found an \LBM\ in around a third of the instances. In contrast, manipulation by removing edges was much less common (e.g., we found an \LBM\ in $0.5\%$ of the instances). We also concentrated on the \MU\ objective, since we do not have efficient algorithms for finding manipulations with \ME\ and \ALO. 
For each graph, we considered every vertex that has the potential to be a manipulator. Specifically, we computed all the $min$-$cut$s, and considered only the vertices that have an edge in a $min$-$cut$ as potential manipulators. Overall, we had $368$ instances. For each instance, we ran our XP algorithm (that finds \LBM, \UBM, \WIM, and \SIM) and the algorithm that finds an \SAM. We also ran a brute-force algorithm, which iterates over all subsets of edges that the manipulator can possibly add. We set a timeout of one hour and terminated each algorithm that did not finish till the timeout.

\begin{figure}[hbpt]
            \centering
            \includegraphics[page=1,width=0.5\textwidth]{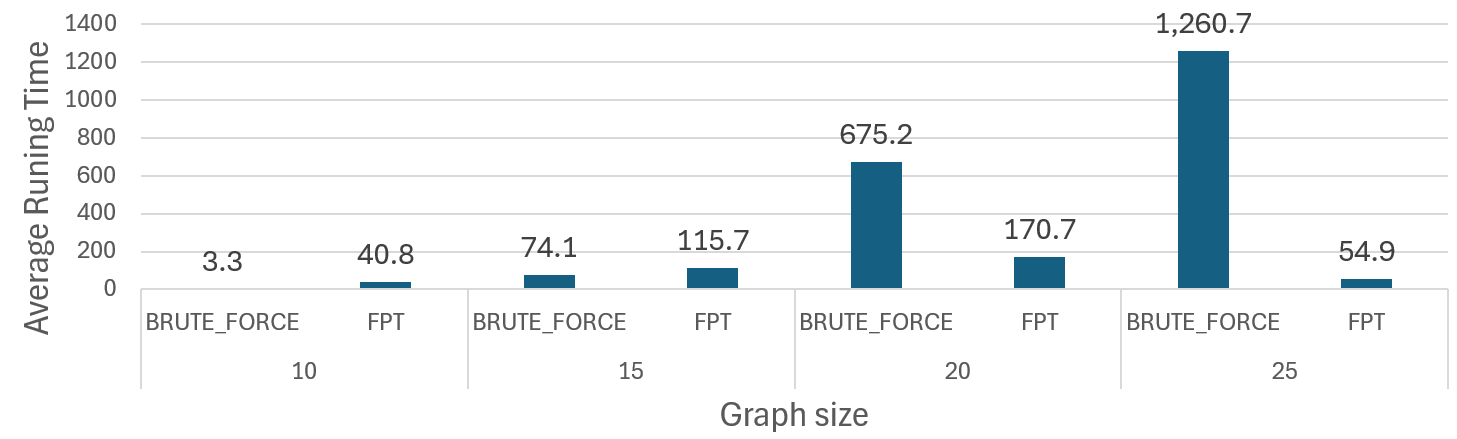}
\caption{Average running time (seconds).}
            \label{fig:directed_graphs_running_time}
         
        \hfill
\end{figure}

%

\paragraph{Frequency of SAM.}
For each graph size, we compared the number of instances in which we found an \LBM\ with the number of instances in which we found an \SAM. Overall, our results indicate that \SAM\ is quite frequent. Specifically, when the graph size was $5$, $10$, $15$, $20$, or $25$, the percentage of instances in which there was an \SAM\ out of the instances in which there was an \LBM\ was $90\%$, $68\%$, $50\%$, $67\%$, and $64\%$, respectively.


\paragraph{Effectiveness of the XP algorithm.}
As expected, the algorithm that finds \SAM\ was very effective: in all of the instances, it finished in less than a second. With a graph size of $5$, the XP and the brute force algorithms also finished in less than a second. The results with the other graph sizes are depicted in Figure \ref{fig:directed_graphs_running_time}. 
Note that we considered only instances in which both algorithms finished their run till the timeout.
The average running time of the brute force algorithm significantly increased as the graph size increased. On the other hand, the XP algorithm was very effective when we increased the graph size. Note that the average running time of the XP with a graph size of $25$ was lower than with graph sizes of $20$ or $15$, due to the timeouts: with a graph size of $25$, the brute force algorithm timed out in $86\%$ of the instances, and thus there were fewer instances in which we tested the algorithms.
Indeed, if we do not remove the instances in which the brute force algorithm timed out, the average running time of the XP algorithm becomes $176$ seconds.
Overall, the brute force algorithm timed out on $0\%$, $49\%$, and $86\%$ of the instances with graph sizes of $15$, $20$, and $25$, respectively. In comparison, the XP algorithm timeout on just $6\%$, $8\%$, and $20\%$ of the instances with the corresponding graph sizes.

As the brute force timed out on  $86\%$ of the instances with graph sizes $25$, we run only the XP algorithm on sampled subgraphs of size $30-100$, that were sampled in same way as above. We observe an increasing trend in the frequency of timeouts, but in a relatively slow rate, see Figure \ref{fig:timeouts_XP}.
\begin{figure}[H]
    \centering
    \includegraphics{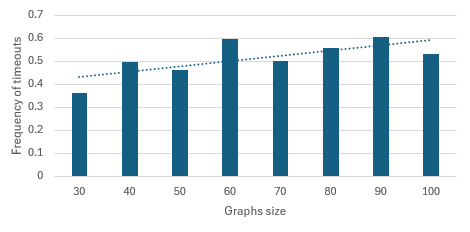}
    \vspace{10pt}
    \caption{Percentage of timeouts of the XP algorithm on graphs with $30-100$ vertices.}
    \label{fig:timeouts_XP}
\end{figure}

         

\section{Conclusions and Future Work}
In this paper, we initiate the study on the complexity of finding manipulation in $k$-coalitional games. We also introduce a new type of manipulation, the socially-aware manipulation.
We showed that both the manipulation and improvement problems are computationally hard with \ME\ and \ALO, even for \SAM. With \MU, we provide a general algorithm that finds an optimal manipulation. The algorithm is XP for LBM, UBM, WIM, and SIM, and it runs in polynomial time for \SAM. Our experiments show that \SAM\ is quite frequent, and demonstrate the effectiveness of our algorithm.
%

The main problem that is still open is classifying the complexity of the manipulation problem for \LBM, \UBM, \WIM, and \SIM\ with \MU. In addition, we showed that the manipulation problem for \SAM\ with \MU\ can be solved efficiently when $k$ is fixed. If $k$ is a parameter, our polynomial-time result becomes an XP-membership result, which naturally leads to the question of whether one can do better (i.e., provide an FPT result).
Finally, it will be interesting to extend our model to the case of multiple manipulators.



\begin{ack}
This research has been partly supported by the Ministry of Science, Technology \& Space (MOST), Israel.
This research has been partially supported by the Israel Science
Foundation under grant 1958/20 and the EU Project
TAILOR under grant 952215.
\end{ack}

\bibliography{bib}

\clearpage
\appendix


\section*{Appendix}
%
    
    
          
         
            

        

\begin{figure}[H]
    \includegraphics[page=1,width=\linewidth]{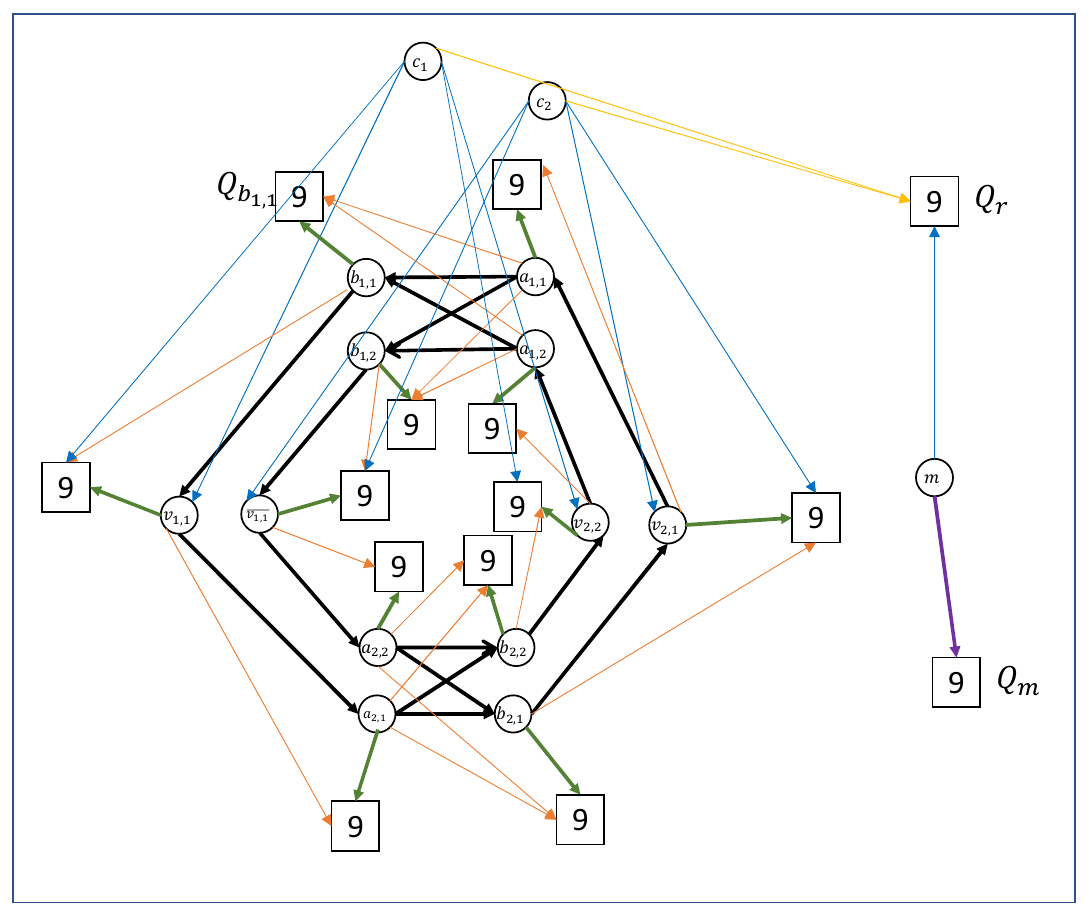}
  
    \caption{An illustration of the graph $G$, from the reduction in the proof of Theorem \ref{thm:sim_me_directed_hard}, for 
   $\mathcal{F} = (x_1 \lor \bar{x_2}) \land (\bar{x_1} \lor x_2)$.
    The color of each arrow indicates the number of edges from the vertex to arbitrary vertices in the clique. Blue represents one edge, orange represents two edges, green represents five edges, yellow represents six edges and purple represents seven edges. The black arrows represent the edges of the ring.
    Note that, to keep the graph size small in the illustration, the example uses a small formula that does not contain $3$ literals in each clause.
     }
       \label{fig:thm1}
\end{figure}
\vspace{0.2cm}

\section{Omitted Proofs}

\subsection{Proof of Theorem \ref{thm:sim_me_directed_hard}}
\begin{reminder}{\ref{thm:sim_me_directed_hard}}
Given a graph $G$, and a manipulator $m^- \in A$, deciding whether any type of manipulation exists when the objective is \ME\ is co-$NP$-hard.
\end{reminder}

As mention in the paper, for the reduction, we use the complementary problem of the NP-complete $3$-SAT problem. And our gadget for the hardness proof is a ring graph, which represents a CNF Boolean formula. It was introduced by \cite{bang2016finding}.

\begin{proof}

Given a formula $\mathcal{F}$, we construct a graph $G=(A,E)$ with a manipulator $m^-$. 

\textbf{The construction of $G$:}
We start by adding a ring graph to $G$ based on $\mathcal{F}$, where $T$ represents the set of ring vertices. For each $t \in T$, we add to $G$  a clique of $9$ vertices, denoted as $Q_t$. 
Each vertex $t\in T$ has the following outgoing edges:
\begin{itemize}
    \item $5$ outgoing edges to $Q_t$
    \item outgoing edges to the next  vertex on the ring (or $2$ next vertices if it is an $a_{i,j}$ vertex that has two next vertices in the ring), following from the ring's definition.
    \item $2$ outgoing edges to vertices within the clique of the next vertex. (or $4$  outgoing edges, if it is an $a_{i,j}$ vertex)
\end{itemize}

In addition to the previous construction, we add a vertex $c_i$ for each clause $C_i \in \mathcal{F}$, which we refer to as clause vertices.
If $x_j\in C_i$, we add an outgoing edge from $c_i$ to one of the occurrences of $x_j$ in the positive path between $s_j$ and $s_{j+1}$, and another edge to one vertex in the corresponding clique of $x_j$ . 
If $\bar{x}_j\in C_i$, we add an outgoing edge from  $c_i$  to one of the occurrences of $\bar{x}_j$ in the negative path between $s_j$ and $s_{j+1}$,  and another edge to one vertex in the corresponding clique of  $\bar{x}_j$.
Note that each literal vertex on the ring has only one incoming edge from a clause vertex. 

We also add an additional clique with $9$ vertices, denoted as $Q_r$. 
Each vertex $c_i$ corresponding to clause $C_i$, has $6$ outgoing edges to vertices of $Q_r$.
Let $G'$ be the subgraph of $G$ containing all the vertices and edges defined up to this point.
To complete the construction, we add to $G$ a clique with $9$ vertices, denoted as $Q_m$. 
Additionally, we add one more vertex, $m^-$, which is defined as the manipulator. $m^-$  has $7$ outgoing edges to vertices of $Q_m$. We also connect $m^-$ to one vertex of $Q_r$ with an outgoing edge. See Figure \ref{fig:thm1} for illustration of the structure of $G$ for  $\mathcal{F} = (x_1 \lor \bar{x_2}) \land (\bar{x_1} \lor x_2)$..

\textbf{Properties of partitions of $G'$.}
 First, we show that the graph $G'$ has a partition to two coalitions, such that each vertex has at least $8$ neighbors if  $\mathcal{F}$ is satisfiable. Conversely, if $\mathcal{F}$ is not satisfiable then every partition of $G'$ has at least one vertex that gets only $6$ neighbors in his coalition.

\begin{enumerate}
    \item Vertices of the cliques: To attain a minimum degree of at least $7$, each vertex must be kept in the same coalition as the other vertices in the clique to which it belongs, as they have no other outgoing edges. Notice that if all vertices of a clique is in the same coalition then the utility of each vertex in this clique is $8$. 
    \item Vertices of the ring: These vertices have $5$ outgoing edges to vertices of their clique, $1$ or $2$ outgoing edges to the next vertices of the ring, and $2$ or $4$ outgoing edges to vertices of the clique of the next vertices of the ring.
    Therefore, to attain at least  $7$ neighbors, every vertex on the ring must be in the same coalition as his corresponding clique, and with at least one of the next vertices.
    This implies that in order to achieve at least $7$ neighbors,  the ring graph can only be partitioned into $2$ cycles, or alternatively, all vertices of the ring can be kept in the same coalition.  Notice that if a complete cycle is in a coalition,  the utility of each vertex in the cycle is at least $8$.
    \item Vertices of the clauses:  These vertices have $6$ outgoing edges to vertices of $Q_r$ and additional $6$ outgoing edges to vertices of the ring. To attain a minimum degree of at least $7$ every vertex $c_i$ must be in the same coalition as $Q_r$, and with at least one of his neighbors from the ring, i.e., a vertex that satisfies this clause. Note that if a clause vertex is in the same coalition with both $Q_r$ and at least one of his neighbors from the ring, denoted by $t$, and $t$ also has a degree of at least $7$, then it implies that the clause vertex has a degree of at least $8$ since it has an edge also to the clique corresponding to $t$.
\end{enumerate}
Hence, if  $\mathcal{F}$ is satisfiable it is possible to achieve a minimum degree of $8$ or more in the partition of $G'$, otherwise,  $\mathcal{F}$ is not satisfiable and the maximum that can be achieved is $6$ in any partition of $G'$.

\textbf{Adding it all together.}
Based on this, we will prove that if $\mathcal{F}$ is satisfiable, then there is no manipulation. Conversely, if $\mathcal{F}$ is unsatisfiable then there is manipulation.

First, we show that if there is a satisfying truth assignment of $\mathcal{F}$ then there is no manipulation. This is because there exists a partition of $G$ where each vertex has at least $8$ neighbors (the maximum possible since some vertices have exactly $8$ friends), the manipulator has only $8$ neighbors and hence in such a partition, it is necessary that $m^-$ gets all of his neighbors in his coalition.
The organizer can partition $G$ into two coalitions using $\tau$  as follows:
The first coalition includes all the vertices in the cycle that corresponds to $\tau$ as well as all clauses vertices $c_i$. It also includes the clique $Q_m$ and the manipulator.
The second coalition includes all the vertices in the cycle that correspond to $\bar{\tau}$.
In this case, the organizer must choose a partition such that $m^-$ will be with $Q_r$ and with all of his friends in $Q_m$, in order for the minimum utility still be $8$. Hence there is no manipulation.

On the other hand, if $\mathcal{F}$  is not satisfiable, then, each partition of $G$ that in which the the vertices of $G'$ are splitted into two different coalitions, achieves a minimum utility of at most $6$. However, the organizer can achieve a minimum degree of $7$, by defining the vertices of $Q_m$ and the manipulator $m^-$ as one coalition and including all vertices of $G'$ in the second coalition.
In this case, the manipulator, $m^-$, has a manipulation by removing edges: $m^-$ can remove two of his outgoing edges to vertices in $Q_m$, and then $m^-$ must be in the same coalition with $Q_r$ in order to achieve $6$ friends in $G(m^-)$. Notice that there are partitions of $G'$ that achieve a minimum degree of $6$; For example, one possible partition  divide the ring into $2$ cycles, placing each cycle in one coalition, and including all $c_i$ with $Q_r$ in one of the coalitions. In such a case, every  $c_i$ vertex has $6$ neighbors, while all other vertices of $G'$ has $8$.
With this manipulation $m^-$ gets all his friends, and it holds that:  $LB(G(m^-),m^-) = 8 > UB(G,m^-) = 7$, hence, there is a \SIM\ (which is also \LBM, \UBM, and \WIM).
\end{proof}

\subsection{Proof of Theorem \ref{thm:improvment_me_directed_hard}}
\begin{reminder}{\ref{thm:improvment_me_directed_hard}}
Given a graph $G$, a manipulator $m^- \in A$, a manipulation type $type$, and a report $E^R(m^-)$, deciding whether $I_{type}(E^R(m^-)) > 0$ when the objective is \ME\ is co-$NP$-hard.
\end{reminder}
\begin{proof}
    The reduction is from the complementary problem of 3-SAT. 
    Given a Boolean CNF formula $\mathcal{F}$, we construct the same graph structure as in the proof of Theorem \ref{thm:sim_me_directed_hard}. 
    In addition, we set $E^R(m^-)$ to be the report in which $m^-$ removes two edges to vertices in $Q_m$. Following the same claims as in the proof of Theorem \ref{thm:sim_me_directed_hard}, we get that $\mathcal{F}$ is unsatisfiable if and only if $I_{type}(E^R(m^-)) > 0$.    
\end{proof}

\subsection{Proof of Theorem \ref{thm:alo_hard_compute}}
\begin{reminder}{\ref{thm:alo_hard_compute}}
Unless $P=NP$, there is no polynomial time algorithm for finding an optimal \LBM\ for a manipulator $m^-$ when the objective is \ALO.
\end{reminder}

For the proof, we need the following lemma:
\begin{lemma}
Each coalition $C$ in an \ALO\ partition contains at least one directed cycle. 
In addition, every vertex in $C$ that is not part of a directed cycle in $C$, is part of a directed path that ends at a directed cycle in $C$.
\label{lemma:alo_contain_cycle}
\end{lemma}
\begin{proof}
Assume by contradiction that there is a coalition in an \ALO\ partition without any directed cycle. Start with an arbitrary vertex, $v$, and consider the longest directed path, $p$, which is started in $v$. Denote by $u$ the last vertex of $p$.
$u$ has no outgoing edge to a vertex in the coalition, but not in $p$. Otherwise, $p$ is not the longest path.  
In addition, $u$ cannot have an outgoing edge to a vertex in $p$. Otherwise, there would be a directed cycle in the coalition. Therefore, $u$ has no neighbors in the coalition, in contradiction to the assumption that this is a coalition in an \ALO\ partition.
\label{thm:alo_np_find_oprimal}
\end{proof}

We can now prove Theorem \ref{thm:alo_hard_compute}. As in the proof of Theorem \ref{thm:sim_me_directed_hard}, our gadget for the hardness proof is a ring graph~\cite{bang2016finding}, which represents a CNF Boolean formula. 
\begin{figure}[hbp]
      \begin{minipage}{\linewidth}
           \centering
       
     \begin{subfigure}{0.45\textwidth}

    \includegraphics[page=1,width=\textwidth]{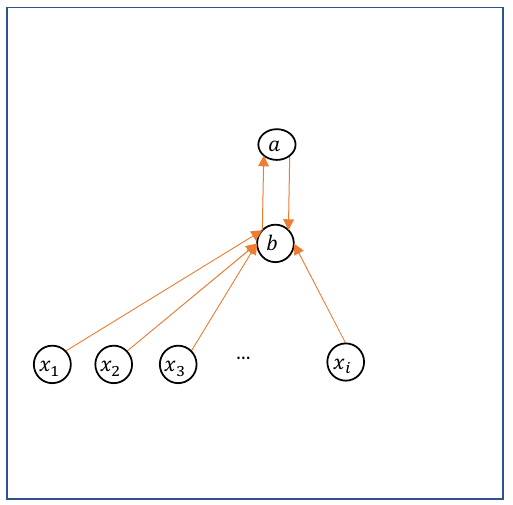}
    
    \end{subfigure}   
       
        \end{minipage}
    \caption{The $D_i$ structure.}\label{fig:D_i}
    
\end{figure}
\vspace{0.2cm}
\begin{proof}
We show a reduction from a variant of the 3-SAT problem, in which each variable occurs in at most $3$ clauses, each literal in at most $2$ clauses, and each clause has $2$ or $3$ literals. This variant is also known as an NP-hard problem \cite{bang2019degree}. 
Given a formula $\mathcal{F}$, we construct a graph $G$ with a  manipulator vertex $m^-$ such that based on the optimal manipulation of $m^-$, we know whether there is a truth assignment that satisfies $\mathcal{F}$. 

\textbf{The construction of $G$:}
We first add to $\mathcal{F}$ an additional clause $C_{m+1} = (a\lor \bar{a})$, where $a$ is a new variable, and denote the resulted formula with $\mathcal{F}'$.
We build a graph $G = (V, E)$ that contains a ring graph with respect to $\mathcal{F}'$.
In addition, we add to $G$ a vertex $c_i$ that corresponds to clause $C_i$, for every  $C_i\in \mathcal{F}'$.  We refer to these vertices as \emph{clause vertices}. 
If $x_j\in C_i$  then we add an outgoing edge from  $c_i$  to one of the occurrences of $x_j$ in the positive path between $s_j$ and $s_{j+1}$. 
If $\bar{x}_j\in C_i$ then we add an outgoing edge from  $c_i$  to one of the occurrences of $\bar{x}_j$ in the negative path between $s_j$ and $s_{j+1}$ (see Figure \ref{fig:ring_clauses} for an illustration).
Note that each literal vertex on the ring has only one incoming edge from a clause vertex. 

For every $i\in[1,m]$ we create $2$ additional vertices $c'_{i,1}$ and $c'_{i,2}$ and connect $c'_{i,1}$ and $c'_{i,2}$ with an outgoing edge to $c_i$.
We also add $5$ vertices $c'_{m+1,1}, \ldots, c'_{m+1,5}$ and the edges $(c'_{m+1,i},  c_{m+1})$ for $1\leq i\leq 5$.

In addition, we add to $G$ the $D_8$ graph described in Figure \ref{fig:D_i} (where $i = 8$).

Finally, we add to $G$ a vertex $m^-$ that represents the manipulator. 
We add the following edges from $m^-$
\begin{itemize}
    \item outgoing edge to every $c_i$, where $i\in [1,m+1]$.
    \item outgoing edge to every $c'_{i,j}$, where $i\in [1,m]$ and $j\in [1,2]$.
    \item outgoing edge to $c'_{m+1,j}$ for $j \in [1,5]$.
    \item outgoing edge to all $x_i$ vertices of $D_8$.
\end{itemize}
 
\textbf{Properties of $G$:}
Before we proceed we emphasize the following properties of $G$: 
\begin{enumerate}
    \item In every \ALO\ partition each  $c'_{i,j}$   is in the same coalition as $c_i$, since $c'_{i, j}$ has only one outgoing edge. (Notice that this implies that if $m^-$ is in the same coalition as $c_i$ then $m^-$ is also in the same coalition as $c'_{i, j}$, and vice versa).
    \label{lemma:c_and_c'_1}
    \item The graph that we build has one cycle in the $D_8$ structure.
    Every other cycle in $G$ is part of the ring graph.
 The ring graph can be partitioned into exactly two disjoint cycles.  Each cycle corresponds to a truth assignment of $\mathcal{F}'$. 
    
    Note that from Lemma \ref{lemma:alo_contain_cycle}, it follows that in every \ALO\ coalition there is at least one cycle and all other vertices that are not in a cycle, are in a directed path to a cycle. 
    Combine this with the fact that the ring can be only partitioned into two disjoint cycles, we get that in every coalition in an \ALO\ partition, that has at least one vertex beside $m^-$ and $D_8$ vertices, there is a cycle of the ring that corresponds to a truth assignment for $\mathcal{F}'$.  
    Moreover, a clause vertex $c_i$ is in a coalition that contains a cycle that corresponds to a truth assignment that satisfies $C_i$, otherwise, $c_i$ is in a coalition that does not contains a cycle that corresponds to a truth assignment that satisfies $C_i$, and hence, $C_i$ does not have any friends in his coalition.
    \label{lemma:c_is_satisfy_1}
    
    
\end{enumerate}
\textbf{Possible \ALO\ partition of $G$:}
Following the properties of $G$ described above, the following \ALO\ partitions of $G$ are the only possible \ALO\ partitions:
\begin{enumerate}
    \item $m^-$ with $D_8$ in the first coalition, and all other vertices in the second coalition.
    \item $D_8$ in the first coalition and all other vertices in the second coalition.
    \item  Split the ring graph into two disjoint cycles and then the first coalition contains one of the cycles of the ring graph. The second coalition contained the other cycle of the ring graph and the vertices of $D_8$.
    The clause vertices are partitioned between the coalitions such that each clause vertex $c_i$ is in a coalition that contains a cycle that corresponds to a truth assignment that satisfies $C_i$.
    In such a case, $m^-$ is either in a coalition that contains at least one clause vertex or in a coalition that contains the vertices of $D_8$. 
\end{enumerate}

Now we show that if there is a truth assignment, $\tau$, that satisfies the formula $\mathcal{F}$, then in the optimal manipulation, $m^-$ keeps only edges to $D_8$, and removes all its other edges.

\begin{Claim}
    There is a truth assignment, $\tau$ that satisfies the formula $\mathcal{F}$, if and only if the manipulator in the optimal manipulation has edges only to $D_8$.
\end{Claim}
\begin{proof}
    
First, we show that if there is a truth assignment that satisfies the formula $\mathcal{F}$, then, the optimal manipulation of $m^-$ is to keep edges only to $D8$, and to remove all other edges.
Note that with a manipulation that keeps only edges to $D_8$, $m^-$ will get at least $8$
friends in his coalition, since all vertices of $D_8$ must be in the same coalition.

With any other manipulation, that keeps edges to some other vertex, $v$, which is not in $D_8$ (even if it is in addition to edges to vertices of $D_8$), we will show that the lower bound for  $m^-$ in $G(m^-)$ is less than $8$. 
The vertex $v$ is one of following vertices:
    \begin{enumerate}
        \item  $v$ is either  $c_{i,j}$ or $c'_{i,j}$ for $1 \leq i \leq m$. The organizer can look at truth assigning  $\tau'$ such that all variables get their value according to $\tau$ except, perhaps, one of the variables that are in $c_i$, so that  $c_i$ will be satisfied by $\bar{\tau}'$.  $\tau'$  does not satisfies at most two clauses (the clause $C_i$ and $C_k$). This is due to the fact that the reduction is from a variation of 3SAT so that each literal is in at most $2$ clauses and therefore changing the value of a single variable can change the value of at most two clauses.
Thus, it can be that one coalition will contain $D_8$, the vertices of $C_{m+1}$ and the cycle corresponding to $\tau'$, and all clause vertices except the vertices of one or maybe $2$ clauses $C_i$ and $C_k$.
The second coalition will contain the second cycle, that corresponds to $\bar{\tau}'$, the vertices of one or maybe $2$ clauses $C_i$ and $C_k$,  and the manipulator $m^-$. In such a partition, $m^-$ gets between $3$ to $6$ friends in his coalition. 
\item $v$ is a vertex of $C_{m+1}$. The organizer can build the partition according to $\tau$, such that one coalition will contain $D_8$,  a cycle in the ring that corresponds to $\tau$, and all clause vertices  except the vertices of  $C_{m+1}$.
The second coalition will contain the second cycle from the ring, that corresponds to $\bar{\tau}$, the vertices of  $C_{m+1}$,  and the manipulator $m^-$. In such a partition $m^-$ gets $6$ friends in his coalition.
    \end{enumerate}

 Now, we turn to show that if the optimal manipulation is to keep edges only to $D_8$, and to remove all other edges, then there is a truth assignment that satisfies the formula $\mathcal{F}$.
Assume that $E^R(m^-)$ that contains only vertices from $D_8$ is an optimal manipulation, and assume by contradiction that $\mathcal{F}$ is not satisfiable. We will show that the manipulation that keeps the edge to a vertex of $C_{m+1}$ achieves higher lower bound for $m^-$.
If $m^-$ has an edge to $D_8$ the organizer can choose a partition such that $m^-$ is with $D_8$ in the first coalition, and all other vertices in the second coalition.  In such a case $m^-$ will get $8$ friends.
But if $E^R(m^-)$ contains edges only to the vertices of $C_{m+1}$, then $m^-$ must be with these vertices in the same coalition, and with one cycle of the ring (in order that the vertices of $C_{m+1}$ get at least one neighbor). 
Since $\mathcal{F}$ is not satisfiable, and a clause vertex gets at least one neighbor in his coalition, only if it is in the same coalition with a cycle of the ring graph that corresponds to a truth assignment that satisfies it, it follows that this coalition must contain the vertices of at least one more clause. 
Hence $m^-$ will get at least $9$ friends. In contradiction to the assumption that manipulation that keeps edges  to $D_8$ is an optimal manipulation.
\end{proof}

\end{proof}

\subsection{Proof of Theorem \ref{thm:improvment_alo_directed_hard}}
%
\begin{reminder}{\ref{thm:improvment_alo_directed_hard}}
Given a graph $G$, a manipulator $m^- \in A$,  and a report $E^R(m^-)$, deciding whether $I_{\LBM}(E^R(m^-)) > 0$  when the objective is \ALO\ is $NP$-hard.  
\end{reminder}
\begin{proof}
We show a reduction from the 3-SAT problem. Given a formula $\mathcal{F}$  in a CNF form, we construct a graph $G$, with a manipulator $m^-$, and a manipulation $E^R(m^-)$ such that there is a truth assignment that satisfies $\mathcal{F}$ if and only if  $E^R(m^-)$ is an \LBM.

\textbf{The construction of $G$:}
We first add to $\mathcal{F}$ an additional clause $C_{m+1} = (a\lor \bar{a})$, where $a$ is a new variable, and denote the resulted formula with  $\mathcal{F}'$.

We build a graph $G=(V,E)$ that contains a ring graph with respect to $\mathcal{F}'$. 
In addition, we add to $G$ a vertex $c_i$ that corresponds to clause $C_i$, for every  $C_i\in \mathcal{F}'$.  We refer to these vertices as clause vertices. 
If $x_j\in C_i$  then we add an outgoing edge from  $c_i$  to one of the occurrences of $x_j$ in the positive path between $s_j$ and $s_{j+1}$. 
If $\bar{x}_j\in C_i$ then we add an outgoing edge from  $c_i$  to one of the occurrences of $\bar{x}_j$ in the negative path between $s_j$ and $s_{j+1}$.  (See Figure \ref{fig:ring_clauses})
Notice that each literal vertex on the ring has only one incoming edge from a clause vertex. 



For every $i\in[1,m]$ we create an additional vertex $c'_i$  and connect $c'_i$ with an outgoing edge to $c_i$.
Notice that we only add $m$ such vertices and the vertex  $c_{m+1}$ currently has no incoming edges.


Finally, we add to $G$ a vertex $m^-$ that represents the manipulator. We connect  $m^-$ with an outgoing edge to every $c_i$, where $i\in [1,m+1]$, and to every $c'_i$, where $i\in [1,m]$.

We set $E^-(m^-) = \{(m^-,c_{m+1})\}$.

\textbf{Properties of $G$:}
Before we proceed, we emphasize some properties of $G$ that follow from its definition. 
\begin{enumerate}
    \item In every \ALO\ partition each $c'_i$  is in the same coalition as $c_i$, since $c'_i$ has only one outgoing edge. (Notice that this implies that if $m$ is in the same coalition as $c_i$ then $m$ is also in the same coalition as $c'_i$, and vice versa).
    \label{lemma:c_and_c'}
    \item Every cycle in $G$ contains only vertices from the ring and the ring itself can be partitioned only into two disjoint cycles.  Each cycle corresponds to a truth assignment of $\mathcal{F}'$. 
    
    (Notice that from Lemma \ref{lemma:alo_contain_cycle} it follows that in every \ALO\ coalition there is at least one cycle. 
    Combine this with the fact that the ring can be partitioned only into two disjoint cycles we get that in every coalition in an \ALO\ partition of $G$ there is one cycle that corresponds to a truth assignment for $\mathcal{F}'$. 
    Moreover, a clause vertex $c_i$ is in a coalition that contains a cycle, which corresponds to a truth assignment that satisfies $C_i$. Otherwise, $c_i$ does not have any friends in his coalition)
    \label{lemma:c_is_satisfy}
    
    
\end{enumerate}

\begin{Claim}
    There is a satisfying truth assignment of $\mathcal{F}$   if and only if $E^R(m) = E(m) \setminus E^-(m^-)$ is an \LBM.
\end{Claim}
First, we show that if there is a satisfying truth assignment of $\mathcal{F}$   then $E^R(m^-)$ is an \LBM.

Let $\tau$ be a satisfying truth assignment of $\mathcal{F}$. 
The organizer can partition $G$ into two coalitions using $\tau$ in the following manner.

The first coalition includes all the vertices in the cycle that correspond  to $\tau$ and the vertices $c_i,c'_i$, where $1\leq i \leq m$. 
The second coalition includes all the vertices in the cycle  that correspond  to $\bar{\tau}$, the vertex $c_{m+1}$ and the manipulator $m^-$.
In such a partition $m^-$ gets only one friend and this is the worst possible scenario for $m^-$.


The manipulation $E^R(m^-)$ of $m^-$ is to remove the edge $(m,c_{m+1})$. Therefore, the above partition is not a valid \ALO\ partition when $m^-$ removes the edge $(m^-,c_{m+1})$. 
In any valid \ALO\ partition of $G(m^-)$ it holds that $m^-$  must be with at least one friend. The only outgoing edges of $m^-$ are to  $c_i$ and $c'_i$, where $i\in [1,m]$.  As we explained earlier  (see  \ref{lemma:c_and_c'}) above) $c_i$ and $c'_i$ are always in the same coalition so $m^-$ must have at least two friends. So after removing the edge $(m,c_{m+1})$  $m^-$ has at least two friends, and hence $E^-(m)$ is a manipulation.

Now, we show that if $E^R(m)$ is an \LBM\ then there is a satisfying truth assignment of $\mathcal{F}$. 
Since $E^R(m)$ is an \LBM\ it follows  that there is a partition $\mathcal{P}$ that was valid in $G$ but is not valid in $G(m)$.
If as a result of the removal of $(m,c_{m+1})$ the partition $\mathcal{P}$ is no longer an \ALO\ partition then it is because $m^-$ has no  friends after the removal of $(m,c_{m+1})$, thus $c_{m+1}$ must be the only one friend of $m^-$ in the same coalition in $\mathcal{P}$. 
So in partition $\mathcal{P}$ all clause vertices $c_i$, where $i\in [1,m]$, must be in the second coalition and as mentioned  above in \ref{lemma:c_is_satisfy} each one of them must be in a coalition that contains a cycle that represent a truth assignment that satisfies the corresponding  clause, so we get that there is a truth assignment that satisfies all clauses of $\mathcal{F}$, and $\mathcal{F}$ is satisfiable. 
\end{proof}

\subsection{Equivalents to Lemma \ref{lemma:remove_regular manipulation_lb}}
Lemma \ref{lemma:remove_regular manipulation_lb} is proved for \LBM\ by a manipulator $m^-$. Here we show equivalent lemmas, for \LBM\ by $m^+$ and for \WIM\ by $m^-$.
\begin{lemma}
If  $E_1^R(m^+)$ is a \LBM\ in which $LB_{\MU}(G_1(m^+), m^+) = x$, then there is a manipulation $E_2^R(m^+)$ such that (i) $E_2^R(m^+)$ respects a partition $\mathcal{P}\in \bar{\mathcal{P}}$ and (ii) $LB_{\MU}(G_2(m^+), m^+) = x$.
\label{lemma:add_regular manipulation_lb}
\end{lemma}
\begin{proof}
 
Given an \LBM\ $E_1^R(m^+)$ in which $LB_{\MU}(G_1(m^+), m^+) = x$, let $\mathcal{P}$ be a partition such that 
$\mathcal{P}\in O_{\MU}(G_1(m^+))$ and $u(m^+,\mathcal{P}) = x$. Let $\mathcal{P} =\{ C_1, C_2 \ldots C_k\}$ and $m^+\in C_1$. By definition of $\Pbar$, $\mathcal{P}\in \Pbar$.
Let $E_2^R(m^+)$ be the report that respects $\mathcal{P}$.
We need to show that $LB_{\MU}(G_2(m^+), m^+) = x$. 
We do so by showing that 
$O_{\MU}(G_2(m^+)) \subseteq O_{\MU}(G_1(m^+))$ and $\mathcal{P} \in O_{\MU}(G_2(m^+))$.

Let $Y$ be the $k$-$cut$ that corresponds to $\mathcal{P}$ in $G$, and let $Y_i$ be the $k$-$cut$ that corresponds to $\mathcal{P}$ in $G_i(m^+)$, where $i\in \{ 1,2\}$. 
%
Given $\mathcal{P}' \not\in O_{\MU}(G_1(m^+))$, let $Y'$ be its corresponding $k$-$cut$ in $G$, and let $Y'_i$ be its corresponding $k$-$cut$ in $G_i(m^+)$, where $i\in \{ 1,2\}$.
Given $\mathcal{P}''\in O_{\MU}(G_1(m^+))$, let $Y''$ be its corresponding $k$-$cut$ in $G$, and let $Y''_i$ be its corresponding $k$-$cut$ in $G_i(m^+)$, where $i\in \{ 1,2\}$.
Let $q=|Y_2|-|Y_1|$, let $q'=|Y'_2|-|Y'_1|$ and let $q''=|Y''_2|-|Y''_1|$.

We begin by proving the following claim:
\begin{Claim}
    $q\leq q'$  and  $q\leq q''$.
    \label{claim:q_q'_m+}
\end{Claim}


\begin{proof}

We divide the set $E_1^+(m^+)\cup E_2^+(m^+)$  into the following three disjoint sets:  
\begin{enumerate}
\item The set $H_1= E_1^+(m^+)\cap E_2^+(m^+)$. 
\item The set $H_2= E_1^+(m^+)\setminus H_1$.
\item The set $H_3= E_2^+(m^+)\setminus H_1$.
\end{enumerate}

The set $H_1$ is contained  both in $G_1(m^+)$ and $G_2(m^+)$, hence it does not affect the value of $q$, $q'$ and $q''$. 

Consider the set $H_2$. 
The set of edges $E_1^+(m^+)\setminus E_2^+(m^+)$ is in $G_1(m^+)$ and not in $G_2(m^+)$.
  Notice also that  $H_2\subseteq \{ (m^+,x)\mid x\in C_2\}$, since otherwise a manipulation that respects $\mathcal{P}$ would also add them.
    So, in the worst case, the values of $q$, $q'$ (and $q''$) are affected in the same way. 
    (Since all of these edges cross the $cut$ of partition $\mathcal{P}$ and are in $G_1$ but not in $G_2$, each one of them decrease the value of $q$,  and if these edges also cross the $cut$ of $\mathcal{P}'$ or $\mathcal{P}''$ it decrease the value of $q'$ or $q''$ at most affected in the same way.)

The set $H_3$ adds edges only to the set of  $G_2(m^+)$.
 $H_3$ is only edges between vertices from $C_1$ and hence only $|Y_2'|$ and $|Y''_2|$) may increase.
 \end{proof}
 We now show that $\mathcal{P} \in O_{\MU}(G_2(m^+))$, and we do so by showing that $Y_2$ is a $min$-$k$-$cut$ in $G_2(m^+)$.
$\mathcal{P}' \not\in O_{\MU}(G_1(m^+))$, and thus  $Y'_1$ is not a $min$-$k$-$cut$ in $G_1(m^+)$. On the other hand, $\mathcal{P} \in O_{\MU}(G_1(m^+))$, and thus $Y_1$ is a $min$-$k$-$cut$ in $G_1(m^+)$. That is,  $|Y_1| < |Y'_1|$. From Claim \ref{claim:q_q'_m+} we have that $q\leq q'$, and thus $|Y_2| < |Y'_2|$. 
Now, $\mathcal{P}'' \in O_{\MU}(G_1(m^+))$ and $\mathcal{P} \in O_{\MU}(G_1(m^+))$. That is, both $Y_1$ and $Y''_1$ are $min$-$k$-$cut$s in $G_1(m^+)$. Therefore, $|Y_1| = |Y''_1|$. From Claim \ref{claim:q_q'_m+} we have that $q\leq q''$ and thus $|Y_2| \leq |Y''_2|$. Overall, $|Y_2| < |Y'_2|$, and  $|Y_2| \leq |Y''_2|$. That is, the size of $Y_2$ is at most the size of any $k$-$cut$. Therefore, $|Y_2|$ is a $min$-$k$-$cut$ in $G_2(m^+)$.

It remains to show that $O_{\MU}(G_2(m^+)) \subseteq O_{\MU}(G_1(m^+))$. Indeed, we showed that  $|Y_2| < |Y'_2|$. Thus, $Y'_2$ is not a $min$-$k$-$cut$ in $G_2(m^+)$. Therefore, $\mathcal{P}' \not\in O_{\MU}(G_2(m^+))$.

Hence $\mathcal{P}$ is a \MU\ partition in $G_2$ and the $UB_{\MU}(G_2(m^+), m^+) = y$.
Since $O_{\MU}(G_2(m^+)) \subseteq O_{\MU}(G_1(m^+))$, it hold that $LB_{\MU}(G_2(m^+), m^+) \geq x$.

\end{proof}

\begin{lemma}
If  $E_1^R(m^-)$ is a \WIM\  in which $LB_{\MU}(G_1(m^-), m^-) = x$, and $UB_{\MU}(G_1(m^-), m^-) = y$, then there is a manipulation $E_2^R(m^-)$ such that (i) $E_2^R(m^-)$ respects a partition $\mathcal{P}\in \bar{\mathcal{P}}$ and (ii)  $LB_{\MU}(G_2(m^-), m^-) \geq x$ and $UB_{\MU}(G_2(m^-), m^-) = y$.
\label{lemma:remove_regular manipulation_wim}
\end{lemma}

\begin{proof}
Given an \WIM\ $E_1^R(m^-)$ in which $LB_{\MU}(G_1(m^-), m^-) = x$, and $UB_{\MU}(G_1(m^-), m^-) = y$.
Let $\mathcal{P}$ be a partition such that 
$\mathcal{P}\in O_{\MU}(G_1(m^-))$ and $u(m^-,\mathcal{P}) = y$. Let $\mathcal{P} =\{ C_1, C_2 \ldots C_k\}$ and $m^-\in C_1$. By definition of $\Pbar$, $\mathcal{P}\in \Pbar$.
Let $E_2^R(m^-)$ be the report that respects $\mathcal{P}$.
We need to show that $LB_{\MU}(G_2(m^-), m^-) = x$ and  $UB_{\MU}(G_2(m^-), m^-) = y$. 
We do so by showing that 
$O_{\MU}(G_2(m^-)) \subseteq O_{\MU}(G_1(m^-))$ and $\mathcal{P} \in O_{\MU}(G_2(m^-))$.

Let $Y$ be the $k$-$cut$ that corresponds to $\mathcal{P}$ in $G$, and let $Y_i$ be the $k$-$cut$ that corresponds to $\mathcal{P}$ in $G_i(m^-)$, where $i\in \{ 1,2\}$. 
%
Given $\mathcal{P}' \not\in O_{\MU}(G_1(m^-))$, let $Y'$ be its corresponding $k$-$cut$ in $G$, and let $Y'_i$ be its corresponding $k$-$cut$ in $G_i(m^-)$, where $i\in \{ 1,2\}$.
Given $\mathcal{P}''\in O_{\MU}(G_1(m^-))$, let $Y''$ be its corresponding $k$-$cut$ in $G$, and let $Y''_i$ be its corresponding $k$-$cut$ in $G_i(m^-)$, where $i\in \{ 1,2\}$.
Let $q=|Y_2|-|Y_1|$, let $q'=|Y'_2|-|Y'_1|$ and let $q''=|Y''_2|-|Y''_1|$.

We begin by proving the following claim:
\begin{Claim}
    $q\leq q'$  and  $q\leq q''$.
    \label{claim:q_q'_wim}
\end{Claim}


\begin{proof}

We divide the set $ E_1^-(m^-)\cup E_2^-(m^-)$  into the following three disjoint sets:  
\begin{enumerate}
\item The set $H_1= E_1^-(m^-)\cap E_2^-(m^-)$. 
\item The set $H_2= E_1^-(m^-)\setminus H_1$.
\item The set $H_3= E_2^-(m^-)\setminus H_1$.
\end{enumerate}

The set $H_1$ is contained in both $E_1^-(m^-)$ and $E_2^-(m^-)$, and thus the edges from $H_1$ are in neither $G_1(m^-)$ nor $G_2(m^-)$. That is, the edges from $H_1$ are not in $Y_i$, $Y'_i$ and $Y''_i$, for $i \in \{1,2\}$.
Since  $q=|Y_2|-|Y_1|$, $q'=|Y'_2|-|Y'_1|$ and $q''=|Y''_2|-|Y''_1|$, then the values of $q$, $q'$ and $q''$ do not depend on the edges from $H_1$. 

Now consider the set $H_2$. The edges from $H_2$ are not in $G_1(m^-)$ but they are in $G_2(m^-)$.
Therefore, the edges of $H_2$ are not in $Y_1$, $Y'_1$, and $Y''_1$.
The manipulation  $E^R_2(m)$ respect $\mathcal{P}$, and thus $H_2\subseteq \{ (m^-,x)\mid x\in C_1\}$.
Therefore, the edges of $H_2$ are not included in any $k$-$cut$ that corresponds to $\mathcal{P}$, i.e., the edges of $H_2$ are also not in $Y_2$. That is, the value of $q$ does not depend on the edges from $H_2$.
However, the edges of $H_2$ may be in $Y'_2$ or in $Y''_2$
We get that each edge from $H_2$ that is in $Y'_2$ or in $Y''_2$ increases the value of $q'$ or $q''$, respectively. 

Finally, consider the set $H_3$. 
The edges from $H_3$ are not in $G_2(m^-)$ but they are in $G_1(m^-)$.
Therefore, the edges of $H_3$ are not in $Y_2$, $Y'_2$, and $Y''_2$.
The manipulation $E^R_2(m)$ respects $\mathcal{P}$, and thus $H_3\subseteq \{ (m^-,x)\mid x\in \{C_i| i \in [2,k]\}\}$.  Therefore, the edges of $H_3$ are in any $k$-$cut$ that corresponds to $\mathcal{P}$, i.e., the edges of $H_3$ are in $Y_1$. That is, each edge from $H_3$ decreases the value of $q$.
In addition, the edges of $H_3$ may be in $Y'_1$ or in $Y''_1$.
We get that each edge from $H_3$ that is in $Y'_1$ or in $Y''_1$ decreases the value of $q'$ or $q''$, respectively. 

Overall, $q\leq q'$ and $q\leq q''$.
\end{proof}

We now show that $\mathcal{P} \in O_{\MU}(G_2(m^-))$, and we do so by showing that $Y_2$ is a $min$-$k$-$cut$ in $G_2(m^-)$.
$\mathcal{P}' \not\in O_{\MU}(G_1(m^-))$, and thus  $Y'_1$ is not a $min$-$k$-$cut$ in $G_1(m^-)$. On the other hand, $\mathcal{P} \in O_{\MU}(G_1(m^-))$, and thus $Y_1$ is a $min$-$k$-$cut$ in $G_1(m^-)$. That is,  $|Y_1| < |Y'_1|$. From Claim \ref{claim:q_q'_wim} we have that $q\leq q'$, and thus $|Y_2| < |Y'_2|$. 
Now, $\mathcal{P}'' \in O_{\MU}(G_1(m^-))$ and $\mathcal{P} \in O_{\MU}(G_1(m^-))$. That is, both $Y_1$ and $Y''_1$ are $min$-$k$-$cut$s in $G_1(m^-)$. Therefore, $|Y_1| = |Y''_1|$. From Claim \ref{claim:q_q'_wim} we have that $q\leq q''$ and thus $|Y_2| \leq |Y''_2|$. Overall, $|Y_2| < |Y'_2|$, and  $|Y_2| \leq |Y''_2|$. That is, the size of $Y_2$ is at most the size of any $k$-$cut$. Therefore, $|Y_2|$ is a $min$-$k$-$cut$ in $G_2(m^-)$.

It remains to show that $O_{\MU}(G_2(m^-)) \subseteq O_{\MU}(G_1(m^-))$. Indeed, we showed that  $|Y_2| < |Y'_2|$. Thus, $Y'_2$ is not a $min$-$k$-$cut$ in $G_2(m^-)$. Therefore, $\mathcal{P}' \not\in O_{\MU}(G_2(m^-))$.

Hence $\mathcal{P}$ is a \MU\ partition in $G_2$ and the $UB_{\MU}(G_2(m^-), m^-) = y$.
Since $O_{\MU}(G_2(m^-)) \subseteq O_{\MU}(G_1(m^-))$, it hold that $LB_{\MU}(G_2(m^-), m^-) \geq x$.

\end{proof}

\subsection{The Full Proof of Theorem \ref{thm:exist-no-harm-man}}
\begin{reminder}{\ref{thm:exist-no-harm-man}}
The susceptibility of \SAM\ is equivalent to the susceptibility of \LBM. Specifically, \MU\ is susceptible to \SAM\ by adding or removing edges, but \ME\ and \ALO\ are susceptible to \SAM\ only by removing edges.

\end{reminder}
\begin{proof}
    
Consider the graph as depicted in Figure\ref{fig:Util_add}. Assume that $k=2$, the organizer's objective is \MU, and $m=m^+$. 
Note that $|min$-$cut(G)|=2$. In addition, $\mathcal{P}_1 = \{\{m\}, \{a,b,c,d\}\} \in O_{\MU}(G)$, and $u(m,\mathcal{P}_1)=0$. Thus, $LB_{\MU}(G,m)=0$. By adding the dashed edge (from $m$ to $a$), the size of $min$-$cut(G(m))$ is still $2$. That is, for each $\mathcal{P} \in O_{\MU}(G(m))$ the size of the corresponding $cut$ in $G(m)$ is $2$.
Since $m$ does not remove any edge, the size of $\mathcal{P}$'s corresponding $cut$ in $G$ is $2$, and thus $\mathcal{P} \in O_{\MU}(G)$. That is, $O_{\MU}(G(m)) \subseteq O_{\MU}(G)$. In addition, every $cut$ in $G(m)$ that contains both $(m,c)$ and $(m,d)$ is of size at least $3$. Therefore, every $min$-$cut$ in $G(m)$ contains at most one edge from $m$. That is, $LB_{\MU}(G(m),m) >0$, and the manipulation is \SAM.

Consider the graph as depicted in Figure\ref{fig:Util_remove}. Assume that $k=2$, the organizer's objective is \MU, and $m=m^-$.
Note that $|min$-$cut(G)|=4$, since $G$ contains the dotted edges. In addition, there are only two partitions in $O_{\MU}(G)$: $\mathcal{P}_1$, in which $m$ is in the same coalition as the vertices of the cliques $A$ and $B$, and $\mathcal{P}_2$, in which $m$ is in the same coalition as the vertices of the cliques $B$ and $C$. Therefore, $u(m,\mathcal{P}_1) = 12$, and $u(m,\mathcal{P}_1) = 14$.
By removing the dotted edges (from $m$ to vertices of $A$), $\mathcal{P}_2$ remains the only partition in $O_{\MU}(G(m))$. That is, the manipulation is \SAM.
 
Consider the graph as depicted in ~(\ref{fig:Egal-remove}).
Assume that $k=2$, the organizer's objective is \ME, and $m=m^-$.
For $\mathcal{P}_1 = \{\{m,d,e,f\}\{a,b,c\}\}$, $\min_{a\in A}u(a,\mathcal{P}_1) =2$. Similarly, for $\mathcal{P}_2 = \{\{a,b,c,m\}, \{d,e,f\}\}$, $\min_{a\in A}u(a,\mathcal{P}_2) =2$. 
Since there are vertices in $G$ that have only two outgoing edges,  $\mathcal{P}_1, \mathcal{P}_2 \in O_{\ME}(G)$, and these are the only partitions in $O_{\ME}(G)$.
Note that $u(m,\mathcal{P}_1) = 2$, and $u(m,\mathcal{P}_2) = 3$.
By removing the dotted edges (from $m$ to $d$ and $e$), $\mathcal{P}_2$ remains the only partition in $O_{\ME}(G(m))$. That is, the manipulation is \SAM.

 

Finally, for \ALO, we will show below that any \LBM\ is also \SAM\ (Theorem \ref{thm:ALO_m-_is_sam}). Therefore, since \ALO\ is susceptible to \LBM\ by removing edges, it is also susceptible to \SAM\ by removing edges.
\end{proof}

\begin{figure}[t]
    \begin{minipage}{1.0\columnwidth}
    \centering  
        \begin{subfigure}{0.145\textwidth}
            \centering
            \includegraphics[page=1,width=\textwidth]{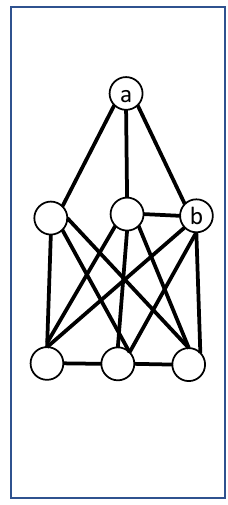}
     \caption{$A_4$  }
            \label{fig:A4_structure}
        \end{subfigure}    
            \begin{subfigure}{0.15\textwidth}
            \centering
            
            \includegraphics[page=1,width=\textwidth]{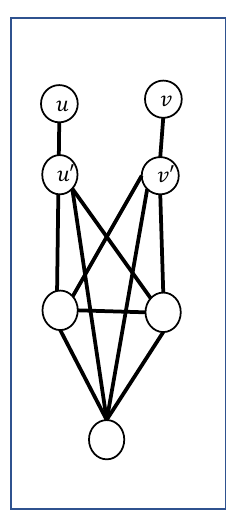}
          \caption{$W_4$}
            \label{fig:W4_structure}
        \end{subfigure}    
        \begin{subfigure}{0.15\textwidth}
            \centering
            
            \includegraphics[page=1,width=\textwidth]{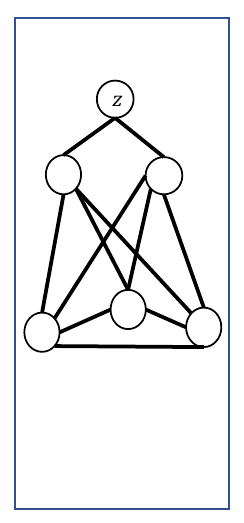}
          \caption{$Z_4$}
            \label{fig:Z4_structure}
        \end{subfigure}    
\vspace{0.4cm}
    \caption{Structures for Theorems~\ref{thorem:max_egal_np_undirected}}
    
    \end{minipage}   
    \vspace{0.4cm}
\end{figure}
\vspace{1.0cm}
\subsection{Proof of Theorem \ref{thorem:sam_max_egal_np_directed}}
\begin{reminder}{\ref{thorem:sam_max_egal_np_directed}}
Given a graph $G$, and a manipulator $m^- \in A$, deciding whether an \SAM\ exists when the objective is \ME\ is co-NP-hard. 
\end{reminder}
\begin{proof}
We show a reduction from the complementary problem of 3-SAT.
Given a Boolean CNF formula $\mathcal{F}$,  we build a graph $G$ with a vertex $m$ that we identify as the manipulator, such that $m$ has a manipulation if and only if  $\mathcal{F}$ is not satisfiable.

Each edge in the construction between vertices $v$ and $u$ is replaced by two directed edges: $(v,u)$ and $(u,v)$, unless stated otherwise.. 
First we add to $G$ the graph $G'$ in the same structure as defined by ~\cite{bang2019degree} in their proof of Theorem 3.6 for the specific case of $k_1 =k_2 =4$. 
They use a ring graph which encodes $\mathcal{F}$ and adds vertex $c_i$ for each clause (as in Figure \ref{fig:ring_clauses}). Add  vertices $u_1, u_2 \ldots u_n$ and connect each $u_i$ by two edges $(u_i,a_{i,1})$ and $(u_i, a_{i,2})$, where $a_{i,1}$ and $a_{i,2}$ are vertices from the ring. 
Similarly, add the vertices $u'_1, u'_2 \ldots u'_n$, and connect each $u'_i$ by two edges $(u'_i,b_{i,1})$ and $(u'_i, b_{i,2})$.

In addition, they added one more vertex, $r$, that is connected to all vertices $H=\{u_1, u_2 \ldots u_n, u'_1, u'_2 \ldots u'_n, c_1, c_2, \ldots, c_m\}$, however the connection is not by edge but via three new vertices - i.e., the $W_4$ structure (depicted in Figure \ref{fig:W4_structure}) such that $r$ is $u$ and $v$ is a vertex from $H$. This manner of connection ensures that in each partition in which every vertex gets at least $4$ neighbors, $r$ will be with all of the vertices in $H$.

For each vertex, $v$, such that $v\in H$ or $v$ is part of the ring, they add to the graph $5$ new vertices. These vertices are in the structure of $Z_4$ (illustration in \ref{fig:Z4_structure}) and the $'z'$ vertex of the structure $Z_4$ is the vertex $v$.

For our setting, we define the $A_4$ graph structure as depicted in Figure \ref{fig:A4_structure} and set the manipulator $m$ to be the $a$ vertex of $A_4$. We also connect $m$ to the vertex $r$ by an outgoing edge from $m$ to $r$.
It follows from the proof of  Theorem 3.6 in \cite{bang2019degree}, when we chose $k_1 =k_2 =4$, that the graph $G'$ has a partition such that each vertex has at least $4$ neighbors if and only if $\mathcal{F}$ is satisfiable. 
Based on this we will prove that there is a manipulation if and only if $\mathcal{F}$ is not satisfiable.
For this we will prove that if $\mathcal{F}$ is satisfiable then there is no manipulation,
and that if $\mathcal{F}$ is  unsatisfiable then there is manipulation.

If $\mathcal{F}$ is satisfiable then $G'$ has a partition such that each vertex has at least $4$ neighbors, and if $A_4$ will be in the coalition with $r$ we will get a minimum degree of $4$, which is the maximum that is possible since there are vertices in the graph that have only $4$ neighbors. 
In this case the organizer must choose a partition such that $m$ will be with $r$ and with all of his neighbors in $A_4$, in order for the minimum degree to still be $4$. Hence there is no manipulation.

On the other hand, if the given formula is not satisfiable, then the maximum that the organizer can achieve is $3$ neighbors for each vertex. Thus, he can choose a partition such that $m$ and $r$ will not be in the same coalition.
Then, $m$ has a \LBM:
\begin{enumerate}
    \item Manipulation by removing edges: $m$ can remove his edge to the vertex $b$ in $A_4$, and then $m$ must be in the same coalition with $r$ in order to achieve a minimum degree of $3$. $m$ must be with all of the vertices in $A_4$, since $m$ must get his $2$ other neighbors from $A_4$ in order to have $3$ neighbors. His neighbors must get $2$ more neighbors from $A_4$ and then $b$, and one more vertex must be with them in order to get $3$ neighbors.
    In this way $m$ gets all his neighbors, hence this is a \LBM.
    \item Manipulation by adding edges: $m$ can add edges to the $c_j$, and then there is a partition that achieves a minimum degree of $4$. In a \ME\ partition of $G(m)$, in order to achieve minimum degree of $4$, the $c_j$ vertices, must be with $r$, and there is at least one clause vertex that is not satisfied, hence $m$ must with it, and then $m$ will get $r$ in his coalition. In addition, $m$ must be with $A_4$ since there is a vertex in $A_4$ that has only $3$ neighbors in addition to $m$. In this way $m$ gets all of his neighbors and hence this is a \LBM.
\end{enumerate}
Note that this \LBM\ is \SAM.
\end{proof}

\subsection{Proof of Theorem \ref{thorem:sam_max_egal_np_directed_improvment}}
\begin{reminder}{\ref{thorem:sam_max_egal_np_directed_improvment}}
Given a graph $G$, a manipulator $m^- \in A$, and a report $E^R(m)$, deciding whether $I_{\SAM}(E^R(m))>0$ when the objective is \ME\ is co-NP-hard. 
\end{reminder}
The reduction is from the complementary problem of 3-SAT. 
    Given a Boolean CNF formula $\mathcal{F}$, we construct the same graph structure as in the proof of Theorem \ref{thorem:sam_max_egal_np_directed}. 
    In addition, we set $E^R(m)$ to the report in which $m$ from the proof of  Theorem \ref{thorem:sam_max_egal_np_directed}. Following the same claims, we get that $\mathcal{F}$ is unsatisfiable if and only if $I_{\SIM}(E^R(m)) > 0$.

\section{An Improved Algorithm for Manipulation By Removing Edges with \MU}
In the paper we presented Algorithm \ref{alg:algorithm_max_util}, which is a general XP algorithm for finding any type of optimal manipulation, with any type of manipulator. Here we present  Algorithm \ref{alg:algorithm_mu_m-}, which is suitable only for $m^-$, and it slightly improves the running time of 
Algorithm \ref{alg:algorithm_max_util}.

The correctness of the algorithm is straightforward, since $m$ should remove only edges that are part of a $k$-$cut$ that $m$ wants that the organizer will choose. In addition $m$ can remove edges only from the wanted $k$-$cut$ in order to achieve a manipulation. Hence he should not remove more edges then he loses in the partition of $G$.  
For each such manipulation we check the improvement and return the optimal manipulation (if exist). 
 As for the running time, note that the number of edges that $m$ loss in the partition of $G$ is bounded by the cost of the $min$-$k$-$cut$, hence in this case, it gets better complexity to check each possible manipulation. i.e., set of edges $E^R(m^-)$  of size at most $N(m^-) - LB(G,m^-)-1$, for each such set it  checks if $E'$ is better than $manip$, and thus it needs to iterate over all the $min$-$k$-$cut$s of $G(m^-)$. Computing all the $min$-$k$-$cut$ takes at most $\mathcal{O}(n^{2k})$ \cite{gupta2019number}. 
    Overall, the running time is $\mathcal{O}(n^{\frac{1}{2}maxSize + 2k)})$, and  $maxSize = 2\cdot mim$-$k$-$cut$.

 \begin{algorithm}
\caption{Compute manipulation for $m^-$}
\label{alg:algorithm_mu_m-}
\textbf{Input}: $G=(A,E) , m^-\in A$\\
\textbf{Output}: $E^R(m^-)$.

\begin{algorithmic}[1] 

\STATE $manip \gets E(m^-)$
\STATE $max-size \gets N(m^-) - LB(G,m^-)-1$
\FORALL{possible report $E^R(m^-)$ of size at most $max-size$}
\IF {$E^R(m^-)$ is better then $manip$}
\STATE $manip \gets E^R(m^-)$
\ENDIF
\ENDFOR
\IF {$manip = E(m^-)$}
\RETURN No manipulation
\ENDIF
\RETURN $manip$

\end{algorithmic}
\end{algorithm}

\section{The Complexity of the Manipulation and the Improvement Problems in Undirected Graphs}
\subsection{Definitions}

\label{def_undirected_graphs}
 When the organizer decides to build an undirected graph, he needs to choose a policy of how to handle inconsistencies, i.e., if  $(b,a) \in E^R(b)$ but $(a,b) \notin E^R(a)$. Clearly, there are two possible options:
\begin{enumerate}
    \item \emph{Weak-edges}: the organizer decides that $(a,b) \in E$ if either $(b,a) \in E^R(b)$ or $(a,b) \in E^R(a)$.
    \item  \emph{Strong-edges}: the organizer decides that $(a,b) \in E$ only when both $(b,a) \in E^R(b)$ and $(a,b) \in E^R(a)$.
\end{enumerate} 
The weak-edges policy is implemented in scenarios involving communication issues, where the organizer acknowledges the possibility that agent $a$ may report on agent $b$ as a friend, but this information may not reach the organizer due to communication issues.

The strong-edges policy is implemented when the organizer recognizes that individuals may have varying perceptions of friendship. Under this policy, a friendship is considered valid only when both parties involved mutually consider each other as friends. 
If $G$ is an undirected graph, then a manipulation by adding edges is relevant only when the organizer uses the weak-edges policy. 
On the other hand, a manipulation by removing edges is relevant only
when the organizer uses the strong-edges policy. Therefore, a manipulation in which edges are both added and deleted is impossible when $G$ is an undirected graph. 
Note that we assume that the manipulator is familiar with how the organizer builds the graph from the reports.

 Table \ref{tab:susceptible2} summarizes the susceptible results from \cite{waxman2021manipulation} for undirected graph in full information setting.

\begin{table}[]
    \centering
    \begin{tabular}{c|c|c}
         & $Add$ & $Remove$\\
         \textbf{Max-Util} & \SIM & \SIM \\
         \textbf{At-Least-1} & SIM & \LBM, no \UBM \\
         \textbf{Max-Egal} & LBM, UBM, no WIM & \SIM 
    \end{tabular}
    \caption{Summary of susceptible results from \cite{waxman2021manipulation} for undirected graphs in full information setting.}
    \label{tab:susceptible2}
\end{table}




\subsection{\ME}
In undirected graphs, \ME\ is susceptible to \LBM\ and to \UBM, but not to \WIM\ by $m^+$. It is also susceptible to \SIM\ by $m^-$  \cite{waxman2021manipulation}.
We show that deciding whether an \LBM\ exists for a given instance is computationally hard. We conjecture that this problem is hard even for the other manipulation types.
\begin{theorem}
Given an undirected graph $G$, and a manipulator $m \in A$, deciding whether an \LBM\ exists when the objective is \ME\ is co-$NP$-hard.  
\label{thorem:max_egal_np_undirected}
\end{theorem}
\begin{proof}
We show a reduction from the complementary problem of 3-SAT.
Given a Boolean CNF formula $\mathcal{F}$,  we build a graph $G$ with a vertex $m$ that we identify as the manipulator, such that $m$ has a manipulation if and only if  $\mathcal{F}$ is not satisfiable.
First we add to $G$ the graph $G'$ in the same structure as defined by ~\cite{bang2019degree} in their proof of Theorem 3.6 for the specific case of $k_1 =k_2 =4$. 
They use a ring graph which encodes $\mathcal{F}$ and adds vertex $c_i$ for each clause (as in Figure \ref{fig:ring_clauses}). Add  vertices $u_1, u_2 \ldots u_n$ and connect each $u_i$ by two edges $(u_i,a_{i,1})$ and $(u_i, a_{i,2})$, where $a_{i,1}$ and $a_{i,2}$ are vertices from the ring. 
Similarly, add the vertices $u'_1, u'_2 \ldots u'_n$, and connect each $u'_i$ by two edges $(u'_i,b_{i,1})$ and $(u'_i, b_{i,2})$.

In addition, they added one more vertex, $r$, that is connected to all vertices $H=\{u_1, u_2 \ldots u_n, u'_1, u'_2 \ldots u'_n, c_1, c_2, \ldots, c_m\}$, however the connection is not by edge but via three new vertices - i.e., the $W_4$ structure (depicted in Figure \ref{fig:W4_structure}) such that $r$ is $u$ and $v$ is a vertex from $H$. This manner of connection ensures that in each partition in which every vertex gets at least $4$ neighbors, $r$ will be with all of the vertices in $H$.

For each vertex, $v$, such that $v\in H$ or $v$ is part of the ring, they add to the graph $5$ new vertices. These vertices are in the structure of $Z_4$ (illustration in \ref{fig:Z4_structure}) and the $'z'$ vertex of the structure $Z_4$ is the vertex $v$.

For our setting, we define the $A_4$ graph structure as depicted in Figure \ref{fig:A4_structure} and set the manipulator $m$ to be the $a$ vertex of $A_4$. We also connect $m$ to the vertex $r$ by an edge.
It follows from the proof of  Theorem 3.6 in \cite{bang2019degree}, when we chose $k_1 =k_2 =4$, that the graph $G'$ has a partition such that each vertex has at least $4$ neighbors if and only if $\mathcal{F}$ is satisfiable. 
Based on this we will prove that there is a manipulation if and only if $\mathcal{F}$ is not satisfiable.
For this we will prove that if $\mathcal{F}$ is satisfiable then there is no manipulation,
and that if $\mathcal{F}$ is  unsatisfiable then there is manipulation.

If $\mathcal{F}$ is satisfiable then $G'$ has a partition such that each vertex has at least $4$ neighbors, and if $A_4$ will be in the coalition with $r$ we will get a minimum degree of $4$, which is the maximum that is possible since there are vertices in the graph that have only $4$ neighbors. 
In this case the organizer must choose a partition such that $m$ will be with $r$ and with all of his neighbors in $A_4$, in order for the minimum degree to still be $4$. Hence, there is no manipulation.

On the other hand, if the given formula is not satisfiable, then the maximum that the organizer can achieve is $3$ neighbors for each vertex. Thus, he can choose a partition such that $m$ and $r$ will not be in the same coalition.
Then, $m$ has an \LBM, which depends on the type of the manipulator:
\begin{enumerate}
    \item Manipulation by removing edges: $m$ can remove his edge to the vertex $b$ in $A_4$, and then $m$ must be in the same coalition with $r$ in order to achieve a minimum degree of $3$. $m$ must be with all of the vertices in $A_4$, since $m$ must get his $2$ other neighbors from $A_4$ in order to have $3$ neighbors. His neighbors must get $2$ more neighbors from $A_4$ and then $b$, and one more vertex must be with them in order to get $3$ neighbors.
    In this way $m$ gets all his neighbors, hence this is a \LBM.
    \item Manipulation by adding edges: $m$ can add edges to the $c_j$, and then there is a partition that achieves a minimum degree of $4$. In a \ME\ partition of $G(m)$, in order to achieve minimum degree of $4$, the $c_j$ vertices, must be with $r$, and there is at least one clause vertex that is not satisfied, hence $m$ must with it, and then $m$ will get $r$ in his coalition. In addition, $m$ must be with $A_4$ since there is a vertex in $A_4$ that has only $3$ neighbors in addition to $m$. In this way $m$ gets all of his neighbors and hence this is an \LBM.
\end{enumerate}
\end{proof}

\subsection{\ALO}

\subsubsection{Manipulation by removing edges}

In undirected graphs, the \ALO\ objective is susceptible to \LBM\ by removing edges \cite{waxman2021manipulation}.
We show that an \LBM\ by removing edges can be computed in polynomial time.

For simplicity, we consider first the case of $k=2$. We later show how to extend our result for the case of $k>2$.  

For $a\in A$, let $N^1(a)$ denote the set of vertices where $a$ is their sole neighbor, that is, $v\in N^1(a)$ if and only if $N(v) = \{a\}$.
Similarly, for $a, b\in A$, let $N^1(a,b) $   be the set of vertices for which $a$ and $b$ are their only neighbors, that is, $v\in N^1(a,b)$ if and only if  $N(v) = \{a,b\}$.

Let $V(a,b)$ be the set of vertices that, in any \ALO\ partition in which  both  $a$ and $b$  are in the same coalition, are  guaranteed to be  part of that coalition. Note, we define that $a , b \in V(a,b)$. 
We prove the following lemma: 
\begin{lemma}
In undirected graphs, for every $a\in A$, it holds that $V(m^-,a) = N^1(a) \cup N^1(m^-) \cup N^1(a,m^-)$.  
Moreover, there exists an \ALO\ partition in which  both  $a$ and $m^-$  are in the same coalition, but every $b \notin  V(m^-,a)$, is in the second coalition.
\label{lemma:ALO_undirected}
\end{lemma}
\begin{proof}
Assume that there is an \ALO\ partition in $G$.  We will prove both : (1) in any \ALO\ partition in which  both  $a$ and $m^-$  are in the same coalition, every vertex  $c \in  N^1(a) \cup N^1(m^-) \cup N^1(a,m^-)$, guaranteed to be  part of that coalition.
    (2) exist an \ALO\ partition in which  both  $a$ and $m^-$  are in the same coalition, but every $b \notin  N^1(a) \cup N^1(m^-) \cup N^1(a,m^-)$, is in the second coalition.
    
    (1) Let $c \in  N^1(a) \cup N^1(m^-) \cup N^1(a,m^-)$. The vertex $c$ has only one or two neighbors and $N(c) \subseteq \{a,m^-\}$. Hence, in every partition $\mathcal{P}$ such that $u(c,\mathcal{P}) > 0$, it holds that if $c \in C_1$, then there exists $v\in N(c)$ such that $v \in C_1$. However, $N(c) \subseteq \{a,m^-\}$ and $m^-$ and $a$ are in the same coalition and hence it holds that $a, m^-$ and $c$ are all in $C_1$.
    Therefore, $c \in V(m^-,a)$ and we get that $N^1(a) \cup N^1(m^-) \cup N^1(a,m^-) \subseteq V(m^-,a)$. 
    (2) Let $b \notin  N^1(a) \cup N^1(m^-) \cup N^1(a,m^-)$. 
    Hence, $b$ has a neighbor $d$, such that $d\not\in \{m^-,a\}$. Since $G$ in an undirected graph, $d$ has $b$ as his neighbor,  and hence $b, d$ can be in the second coalition without $m^-$ and $a$.  And we get that $b\notin V(m^-,a)$.
    Moreover, since it holds for every $b\notin  N^1(a) \cup N^1(m^-) \cup N^1(a,m^-)$ we gets that all theses vertices can be together in the second coalition without $m^-$ and $a$.
\end{proof}





\begin{algorithm}
\caption{LBM with the \ALO\ objective in undirected graphs}
\label{alg:algorithm_ALO}
\begin{algorithmic}[1] 
\REQUIRE $G=(A,E) , m^-\in A$\\
\ENSURE $E^R(m^-)$
\IF { $N^1(m^-) \neq \emptyset$} 
\RETURN No manipulation
\ENDIF
\STATE $maxV \gets -1$, $maxA \gets m^-$
\STATE $minV \gets |A|+1$
\FORALL{$a\in N(m^-)$}
\STATE compute $V(m^-,a)$
\IF {$|N(m^-, V(m^-,a))| \geq maxV$ and $|V(m^-,a)|<|A|$}  
\STATE $maxV \gets |N(m^-, V(m^-,a))|$
\STATE $maxA \gets a$
\ENDIF
\IF {$minV \geq |N(m^-, V(m^-,a))|$}
\STATE $minV \gets |N(m^-, V(m^-,a))|$
\ENDIF
\ENDFOR

\IF {$minV = maxV$}
\RETURN No manipulation
\ENDIF
\RETURN $E^R(m^-) = \{(m^-,maxA)\}$ 


\end{algorithmic}
\end{algorithm}

Our algorithm, Algorithm \ref{alg:algorithm_ALO}, works as follows. It first searches for a vertex $a\in N(m^-)$ that has the largest number of  $m^-$ neighbors in the set $V(m^-,a)$. 
If such a vertex is found, then $m^-$ reports only the edge $(m^-,a)$. 
It might be that for every $a\in N(m^-)$, the number of  $m^-$ neighbors in the set $V(m^-,a)$ is the same.
In such a case, the algorithm reports that there is no manipulation.

We now show the correctness of the algorithm. Intuitively, 
if an \ALO\ partition exists in $G(m^-)$, then it guarantees that $m$ obtains both vertex $a$ and all the vertices in the set $V(m^-,a)$.
Note that, if there is no \ALO\ partition in $G$, the manipulator $m^-$ cannot perform a manipulation since he can only remove edges.  Consequently, in $G(m^-)$, there still will not be an \ALO\ partition.
\begin{theorem}
In undirected graphs, Algorithm \ref{alg:algorithm_ALO} finds an optimal \LBM\ by $m^-$ with the \ALO\ objective in polynomial time.
\label{thm:ALO_remove}

\end{theorem}
\begin{proof}
If $m^-$ has a neighbor $a$ such that $a \in N^1(m^-)$, and $m^-$ does not report on the edge $(m^-,a)$, it would result in a graph $G(m^-)$ such that $O_{\ALO}(G(m^-)) = \emptyset$.
Hence, $m^-$ must report on the edge $(m^-,a)$. In such a case, a  partition $C_1,C_2$ such that $C_1 = \{m^-\} \cup N^1(m^-)$ and $C_2 = A\setminus C_1$ is in $O_{\ALO}(G(m^-))$. In addition, in every \ALO\ partition in $G$  all vertices from $N^1(m^-)$ were in the same coalition as $m^-$. Therefore, the $LB(G,m^-)=LB(G(m^-),m^-)$ and there is no manipulation.

Consider now the case that $N^1(m^-) = \emptyset$, Algorithm \ref{alg:algorithm_ALO}, compute for each $a\in N(m^-)$ the set $V(m^-,a)$, and chose to report on only one edge to the vertex that the number of  $m^-$ neighbors in the set $V(m^-,a)$ is highest.

As proven in Lemma \ref{lemma:ALO_undirected}, if there is an \ALO\ partition in $G(m^-)$, it is guaranteed that $m^-$ will be with all vertices of $V(m^-, maxA)$, but all other vertices have neighbors in $A\setminus V(m^-,maxA)$ and hence can be in the second coalition.
In Algorithm \ref{alg:algorithm_ALO}, $m^-$ verifies that $V(m^-,a) \neq A$, ensuring that $O_{\ALO}(G(m^-))$ is not empty.  
In this way, Algorithm \ref{alg:algorithm_ALO}, guarantees a maximum improvement of the lower bound.

Algorithm \ref{alg:algorithm_ALO} operates in polynomial time since the computation of $V(m^-,a)$ is a polynomial-time operation and it is performed only once for each vertex.
\end{proof}
The algorithm can be extended for any $k>2$, as follows.
The set $V(m^-,a)$ is computed in the same way. 
However,  instead of verifying that $|V(m^-,a)| < |A|$, in line $8$ of Algorithm \ref{alg:algorithm_ALO}, the manipulator needs to ensure the existence of an \ALO\ $(k-1)$-partition of the rest of the graph. This is equivalent to the maximum matching problem, which can be solved in polynomial time (e.g., using the algorithm from \cite{micali1980v}). 

\subsubsection{Manipulation by adding edges}

In undirected graphs, the \ALO\ objective is susceptible to \SIM\ by adding edges \cite{waxman2021manipulation}.
Note that when the objective is \ALO,  since $m^+$ can only add edges, every \ALO\ partition in $G$ is also \ALO\ partition in $G(m^+)$. So $m^+$ can only transform non-\ALO\ partitions into \ALO\ partitions. We will prove that a manipulation is possible only when there was no \ALO\ partition initially, and the manipulation results in a graph with an \ALO\ partition in which $m^+$ must get at least one real neighbor.
\begin{lemma}
\sloppy{By $m^+$, the only possible manipulation is when there was no \ALO\ partition in $G$, but there is an \ALO\ partition in $G(m^+)$, i.e., $O_{\ALO}(G) = \emptyset$ but $O_{\ALO}(G(m^+)) \not= \emptyset$.} 
\end{lemma}
\begin{proof}
    Assume, by contradiction, that there exists a graph $G = (V,E)$, a manipulator $m^+ \in V$, and an organizer with the \ALO\ objective, such that there is an \ALO\ partition in $G$, but there is a manipulation such that $LB(G,m^+) < LB(G(m^+),m^+)$ or $UB(G,m^+) < UB(G(m^+),m^+)$.
    First, $LB(G,m^+) < LB(G(m),m^+)$, is not possible since every \ALO\ partition in $G$ remains an \ALO\ partition in $G(m^+)$.  

    Now, assume, by contradiction, that there exists a manipulation $E^+(m^+)$ such that $UB(G,m^+) < UB(G(m),m^+)$, and let $\mathcal{P}=(C_1,C_2)$ be the partition with this new upper bound in $G(m^+)$, where $m^+\in C_1$.
    Let $\mathcal{P}'=(C'_1, C'_2)$ denote the partition that is identical to $\mathcal{P}$ except for the relocation of each vertex  $a\in C'_1$ that $m^+$ add an edge to him, and doesn't have other edges in $C'_1$ to  $C'_2$. The partition $\mathcal{P}'$ is an \ALO\ partition in $G$. (Since there is \ALO\ partition in $G$, it implies that each vertex has at least one neighbor in $G$,  and if a vertex does not have a real neighbor in $C'_1$ it implies that he has in $C'_2$). 
    However, we encounter a contradiction when comparing the utility function $u(m^+, \mathcal{P}')$ = $u(m^+, \mathcal{P})$ as it contradicts the premise that $E^+(m^+)$ increases the upper bound for $m^+$.
    
\end{proof}

\begin{lemma}
    If there is no \ALO\ partition, but each vertex has at least one edge, then no manipulation exists for $m^+$.
\end{lemma}
\begin{proof}
    If the organizer selects an \ALO\ partition in which $m^+$  is placed within a coalition that has only fake neighbors, i.e., vertices that $m^+$ added edges to them, his actual utility will be $0$, and this situation does not qualify as manipulation. 
    Otherwise, if the manipulator $m^+$ has at least one real neighbor within his coalition, denoted by $C_1$ the coalition of $m^+$.
    For each vertex to which $m^+$ added an edge and lacks additional neighbors within $C_1$, has a neighbor in another coalition and can  be relocated to another coalition without reducing $m^+$'s utility. Consequently, the added edges do not provide any advantage to $m^+$.
\end{proof}
Hence, we can conclude this and get the following way of solving the SIM problem.
\begin{theorem}
    In undirected graphs, it is possible to find an \SIM\ by $m^+$  with the \ALO\ objective in polynomial time.
    \label{thm:alo_+_poly}
\end{theorem}
\begin{proof}
In undirected graphs, when the objective of the organizer is \ALO, the only feasible manipulation for $m^+$ is to add edges to vertices that have no neighbors. This manipulation works if (1) $m^+$ has at least one neighbor, denoted as $a$, such that $N^1(a) = {m^+}$ (2) the resulting graph has \ALO\ partition. 
The result of the two previous lemmas is that manipulation in this case can only be accomplished by adding edges to vertices that have no neighbors.
    
    Assume that $m^+$ adds edges to vertices that have no neighbors, and there exists an \ALO\ partition in $G(m^+)$.
    
    We will prove the theorem by considering different scenarios of manipulation:
    
    (1) If $m^+$ has at least one neighbor $a$ with $N^1(a) = \{m^+\}$, then $m^+$ gets these neighbors in every \ALO\ partition, and after his manipulation, there is an \ALO\ partition.
    
    (2) If $m^+$ lacks a neighbor $a$ with $N^1(a) = \{m^+\}$, there is \ALO\ partition such that $m^+$  end up with only fake neighbors, resulting in a real utility of $0$, when all his real neighbors are in another coalition.

    Note that the computation of $N^1(a)$ and the check if there is a vertex with no neighbors can be done in polynomial time. 
\end{proof}
\begin{theorem}
    In undirected graphs, it is possible to find a \UBM\ by $m^+$  with the \ALO\ objective in polynomial time.
\end{theorem}
It is also possible to find a \UBM\ by $m^+$  with the \ALO\ objective in polynomial time, in a relatively similar way. 
However, that if $O_{\ALO}(G,m)=\emptyset$ and there is a neighbor of $m$, denote by $a$, that can be with $m$ in the same coalition, and still, every other vertex will have at least one neighbor  (i.e., $|V(m^+,a)| < |A|$). Then we can say that there is a \UBM, since there was no \ALO\ partition in $G$, but in $G(m)$ there are \ALO\ partitions and in some of them $m$ gets at least one neighbor.

\subsection{\MU} 
In undirected graphs, the \MU\ objective is susceptible to \SIM. Fortunately, Algorithm \ref{alg:algorithm_max_util} can also solve the manipulation problem in undirected graphs.
\end{document}